\documentclass[11pt]{article}

\title{$\widetilde{O}(n^{1/3})$-Space Algorithm for the Grid Graph Reachability Problem}

\author{
Ryo Ashida \thanks{Department of Mathematical and Computing Science, Tokyo Institute of Technology. E-mail: ashida1@is.titech.ac.jp}
\and
Kotaro Nakagawa \thanks{JMA SYSTEMS Corporation. E-mail: kootaroonakagawa@gmail.com}
}

\usepackage{algorithm, algorithmic,enumitem,graphicx, amsthm, amsfonts, amsmath, geometry}

\newcommand{\eps}{\varepsilon}
\newcommand{\tldO}{\widetilde{O}}
\newcommand{\newwd}[1]{{\em#1}}

\newtheorem{definition}{Definition}
\newtheorem{theorem}{Theorem}
\newtheorem{lemma}{Lemma}

\newcommand{\rimV}{V^{\rm rim}_{0}}
\newcommand{\blcG}{G_{0}}
\newcommand{\blcV}{V_{0}}
\newcommand{\blcE}{E_{0}}
\newcommand{\blcC}{C_{0}}

\newcommand{\cirblcG}{\blcG^{\rm cir}}
\newcommand{\cirblcV}{\blcV^{\rm cir}}
\newcommand{\cirblcE}{\blcE^{\rm cir}}
\newcommand{\cC}{C^{\rm cl}}
\newcommand{\aC}{C^{\rm acl}}
\newcommand{\gadG}{\widetilde{G}}
\newcommand{\gadV}{\widetilde{V}}
\newcommand{\gadE}{\widetilde{E}}
\newcommand{\gadK}{\widetilde{K}}
\newcommand{\gadL}{\widetilde{L}}

\newcommand{\gadVout}{\widetilde{V}^{\rm out}}

\newcommand{\gadplG}{\widetilde{G}_{\sf p}}
\newcommand{\gadplV}{\widetilde{V}_{\sf p}}
\newcommand{\gadplVin}{\widetilde{V}_{\sf p}^{\rm in}}
\newcommand{\gadplVout}{\widetilde{V}_{\sf p}^{\rm out}}
\newcommand{\gadplE}{\widetilde{E}_{\sf p}}
\newcommand{\gadplK}{\widetilde{K}_{\sf p}}
\newcommand{\gadplL}{\widetilde{L}_{\sf p}}

\newcommand{\lin}{\ell_{in}}
\newcommand{\lout}{\ell_{out}}

\date{}

\begin{document}

\maketitle

\begin{abstract}
The directed graph reachability problem takes
as input an $n$-vertex directed graph $G=(V,E)$,
and two distinguished vertices $s$ and $t$.
The problem is to determine whether there exists a path
from $s$ to $t$ in $G$.
This is a canonical complete problem for class NL.
Asano et al. proposed
 an $\widetilde{O}(\sqrt{n})$ space\footnote{In this paper ``$\tldO(s(n))$ space'' means $O(s(n))$ words intuitively and precisely $O(s(n)\log n)$ space.} and polynomial time algorithm
 for the directed grid and planar graph reachability problem.
The main result of this paper is to show  that
the directed graph reachability problem restricted to grid graphs can be solved in polynomial time using only $\widetilde{O}(n^{1/3})$ space.
 \end{abstract}

\section{Introduction}
The graph reachability problem,
for a graph $G=(V,E)$ and two distinct vertices $s,t\in V$,
is to determine whether there exists a path from $s$ to $t$.
This problem characterizes many important complexity classes.
The {\it directed graph reachability problem} is a canonical complete problem
for the nondeterministic log-space class, NL.
Reingold showed that the {\it undirected graph reachability problem} characterizes the deterministic log-space class, L\cite{reingold2008undirected}.
As with P vs. NP problem,
whether L=NL or not is a major open problem.
This problem is equivalent to whether the directed graph reachability problem is solvable in deterministic log-space.
There exist two fundamental solutions for the directed graph reachability problem,
breadth first search, denoted as BFS, and Savitch's algorithm.
BFS runs in $O(n)$ space and $O(m)$ time,
where $n$ and $m$ are the number of vertices and edges, respectively.
For Savitch's algorithm,
we use only $O(\log^2 n)$ space
but require $\Theta(n^{\log n})$ time.
BFS needs short time but large space.
Savitch's algorithm uses small space but super polynomial time.
A natural question is whether we can make an efficient deterministic algorithm in both space and time for the directed graph reachability problem.
In particular,
Wigderson proposed a problem that does there exist an algorithm for the directed graph reachability problem that uses polynomial time and $O(n^{\eps})$ space, for some $\eps < 1$? \cite{wigderson1992complexity},
and this question is still open.
The best known polynomial time algorithm, shown by Barns, Buss, Ruzzo and Schieber, uses $O(n/2^{\sqrt{\log n}})$ space  \cite{barnes1998sublinear}.

For some restricted graph classes,
better results are known.
Stolee and Vinodchandran showed that
for any $0 < \eps < 1$,
the reachability problem for directed acyclic graph with $O(n^{\eps})$ sources and 
embedded on a surface with $O(n^{\eps})$ genus can be solved in polynomial time and $O(n^{\eps})$ space \cite{stolee2012space}.
A natural and important restricted graph class is the class of planar graphs.
The {\it planar graph reachability problem} is
in the unambiguous log-space class, UL \cite{bourke2009directed}, which is a subclass of NL.
Imai et al. gave an algorithm using $O(n^{1/2+\eps})$ space and polynomial time for the planar graph reachability problem \cite{asano2011memory,imai2013n12+}.
Moreover
Asano et al. devised a efficient way to control the recursion, and
proposed a polynomial time and $\tldO(\sqrt{n})$ space algorithm  for the planar graph reachability problem \cite{asano2014widetilde}.
In this paper,
we focus on the {\it grid graph reachability problem}, where grid graphs are special cases of planar graphs.
Allender et al. showed  the planar graph reachability problem is log-space reducible to the grid graph reachability problem \cite{allender2009planar}.
By using the algorithm of Asano et al.,
we can solve the grid graph reachability problem
in $\tldO(\sqrt{n})$ space and polynomial time.
The main result of this paper is to show an $\tldO(n^{1/3})$ space and polynomial time algorithm for the directed grid graph reachability problem.

\begin{theorem}[\cite{asano2014widetilde}]\label{thm:planar-reach}
	There exists an algorithm that decides directed planar graph reachability in polynomial time and $\tldO(\sqrt{n})$ space.
	$($We refer to this algorithm by {\sf PlanarReach} in this paper.$)$
\end{theorem}

\section{Preliminaries and an Outline of the Algorithm}
We will use the standard notions and notations for algorithms, complexity measures, and graphs without defining them.
We consider mainly directed graphs, and a graph is assumed to be a directed graph unless it is specified as a undirected graph. 
Throughout this paper, for any set $X$, $|X|$ denotes the number of elements in $X$.
We refer to the maximum and minimum elements of $X$
as $\max X$ and $\min X$, respectively.
Consider any directed graph $G = (V, E)$.
For any $u, v \in V$, a directed edge $e$ from $u$ to $v$ is denoted as $e = (u, v)$;
on the other hand,
the tail $u$ and the head $v$ of $e$ are denoted as $t(e)$ and $h(e)$, respectively.
For any $U\subseteq V$, let $G[U]$ denote the subgraph of $G$ induced by $U$.

Recall that a grid graph is a graph whose vertices are located on grid points, and whose vertices are adjacent only to their immediate horizontal or vertical neighbors.
We refer to a vertex on the boundary of a grid graph as a \newwd{rim vertex}.
For any grid graph $G$,
we denote the set of the rim vertices of $G$ as $R_{G}$.

\medskip\noindent
\underline{Computational Model}

\medskip\noindent
For discussing sublinear-space algorithms formally,
we use the standard multi-tape Turing machine model.
A multi-tape Turing machine consists of
a read-only input tape,
a write-only output tape,
and a constant number of work tapes.
The space complexity of this Turing machine is measured
by the total number of cells that can be used as its work tapes.

For the sake of explanation,
we will follow a standard convention and give a sublinear-space algorithm
by a sequence of constant number of sublinear-space subroutines
$A_1,\ldots, A_k$
such that each $A_i$ computes, from its given input,
some output that is passed to $A_{i+1}$ as an input.
Note that
some of these outputs cannot be stored in a sublinear-size work tape;
nevertheless,
there is a standard way
to design a sublinear-space algorithm based on these subroutines.
The key idea is
to compute intermediate inputs every time when they are necessary.
For example,
while computing $A_i$,
when it is necessary to see the $j$th bit of the input to $A_i$,
simply execute $A_{i-1}$ (from the beginning)
until it yields the desired $j$th bit on its work tape,
and then resume the computation of $A_i$ using this obtained bit.
It is easy to see that
this computation can be executed in sublinear-space.
Furthermore,
while a large amount of extra computation time is needed,
we can show that
the total running time can be polynomially bounded
if all subroutines run in polynomial-time.

\medskip\noindent
\underline{Outline of the Algorithm}

\medskip\noindent
We show the outline of our algorithm. 
Our algorithm uses the algorithm {\sf PlanarReach} for the planar graph reachability.
We assume both $\sqrt{n}$ and $n^{1/3}$ are integers for simplicity.
Let $G$ be an input $\sqrt{n}\times\sqrt{n}$ grid graph with $n$ vertices.
\begin{enumerate}
	\item Separate $G$ into $n^{1/3}\times n^{1/3}$ small grid graphs, or ``blocks''.
	There are $n^{1/3}$ blocks, and each block contains $n^{2/3}$ vertices.
	\item Transform each block $B$ into a special planar graph, ``gadget graph'', with $O(n^{1/3})$ vertices. The reachability among the vertices in $R_{B}$ should be unchanged. The total number of vertices in all blocks becomes $O(n^{2/3})$.
	\item We apply the algorithm {\sf PlanarReach} to the transformed graph of size $O(n^{2/3})$, then the reachability is computable in $\tldO\left(\sqrt{n^{2/3}}\right) = \tldO(n^{1/3})$ space.
\end{enumerate}

In step 1 and 2, we reduce the number of vertices in the graph $G$ while keeping the reachability between the rim vertices of each block so that we can solve the reachability problem of the original graph.
Then to this transformed graph we apply {\sf PlanarReach} in step 3, which runs in $\tldO(n^{1/3})$ space.

\begin{theorem}\label{thm:main}
There exists an algorithm that computes the grid graph reachability in polynomial-time and $\tldO(n^{1/3})$ space.
\end{theorem}

The start vertex $s$ (resp., the end vertex $t$) may not be on the rim of any block.
In such a situation,
we make an additional block so that $s$ (resp., $t$) would be on the rim of the block.
This operation would not increase the time and space complexity.
In this paper, we assume that $s$ (resp., $t$) is on the rim of some block.

\section{Graph Transformation}
In this section,
we explain an algorithm that modifies each block
and analyze time and space complexity of the algorithm.
Throughout this section,
we let a directed graph $\blcG=(\blcV,\blcE)$
denote a block of the input grid graph,
and let $\rimV$ denote the set of its rim vertices.
We use $N$ to denote the number of vertices of the input grid graph and
$n$  to denote $|\rimV|$, which is $O(N^{1/3})$;
note, on the other hand,
that we have $|\blcV|=O(n^2)=O(N^{2/3})$.
Our task is to transform this $\blcG$ to a plane ``gadget graph'',
an augmented plane graph, $\gadplG$ with $O(n) = O(N^{1/3})$ vertices including $\rimV$
so that
the reachability among vertices in $\rimV$ on $\blcG$
remains the same on $\gadplG$.

There are two steps for this transformation.
We first transform $\blcG$ to a circle graph $\cirblcG$,
and then obtain $\gadplG$ from the circle graph.

\subsection{Circle Graph}

\begin{figure}
\begin{center}
\includegraphics[width= 14 cm]{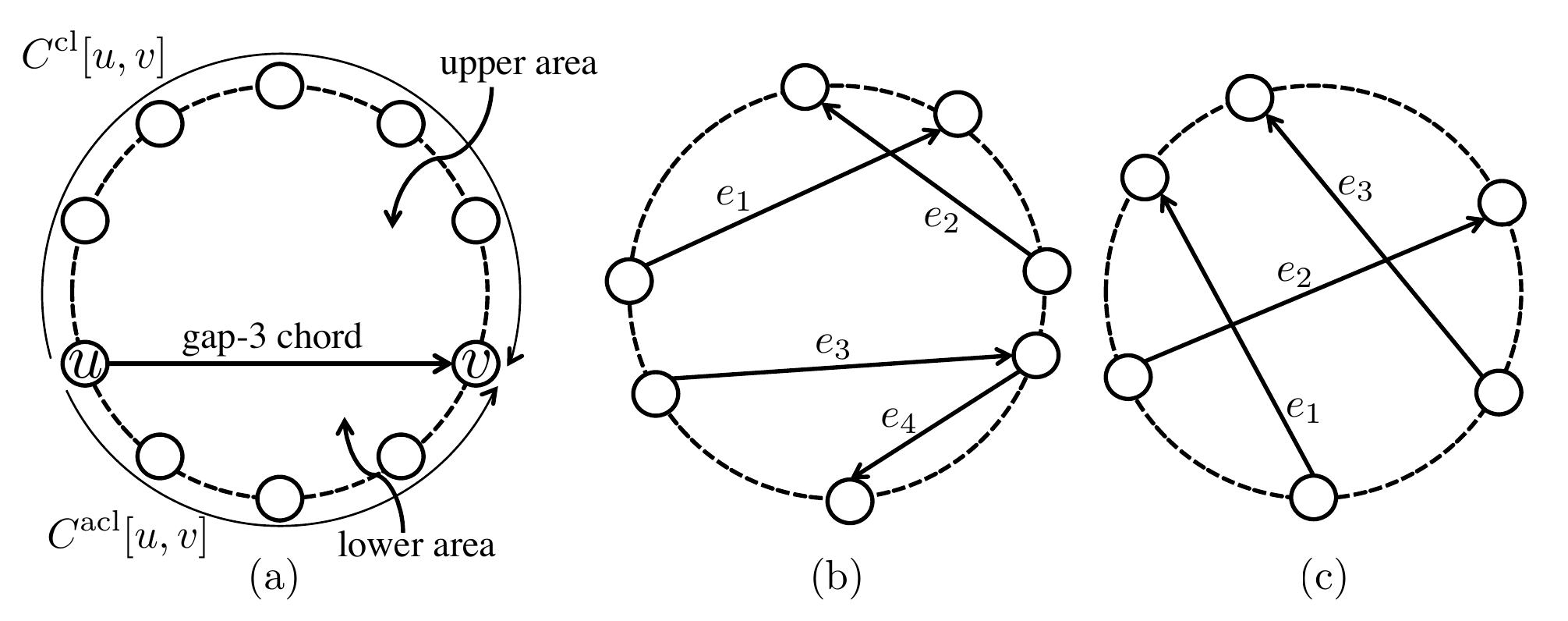}
\caption{An example of the notions on chords. (a) a figure showing a chord, arcs, a lower area, an upper area,
(b) a figure showing crossing chords ($e_1$ and $e_2$) and semi-crossing chords ($e_3$ and $e_4$) and (c) separating chords ($e_3$ separates $e_1$ and $e_2$).}
\label{fig:chord_arc}
\end{center}
\end{figure}

We introduce the notion of ``circle graph''.
A \newwd{circle graph} is
a graph embedded on the plane
so that all its vertices are placed on a cycle
and all its edges are drawn inside of the cycle.
Note that
a circle graph
may not have an edge
between a pair of adjacent vertices on the cycle.
We introduce some basic notions on circle graphs.
Consider any circle graph $G=(V,E)$,
and let $C$ be a cycle on which all vertices of $V$ are placed.
For any $u,v\in V$,
a \newwd{clockwise tour} (resp., \newwd{anti-clockwise tour})
is a part of the cycle $C$ from $u$ to $v$ in a clockwise direction
(resp., in an anti-clockwise direction).
We use $\cC[u,v]$ (resp., $\aC[u,v]$) to denote this tour (Figure~\ref{fig:chord_arc}(a)).
When we would like to specify the graph $G$,
we use $\cC_{G}[u,v]$ (resp., $\aC_{G}[u,v]$).
The tour $\cC[u,v]$,
for example,
can be expressed canonically
as a sequence of vertices $(v_1,\ldots,v_k)$
such that $v_1=u$, $v_k=v$, and $v_2,\ldots,v_{k-1}$ are all vertices
visited along the cycle $C$ clockwise.
We use $\cC(u, v)$ and $\cC[u, v)$ (resp., $\aC(u, v)$ and $\aC[u, v)$) to denote the sub-sequences
$(v_2, \ldots, v_{k-1})$ and $(v_1,\ldots, v_{k-1})$ respectively.
Note here that
it is not necessary that
$G$ has an edge between adjacent vertices in such a tour.
The length of the tour is simply the number of vertices on the tour.
An edge $(u,v)$ of $G$ is called a \newwd{chord}
if $u$ and $v$ are not adjacent on the cycle $C$.
For any chord $(u,v)$,
we may consider two arcs,
namely,
$\cC[u,v]$ and $\aC[u,v]$;
but in the following,
we will simply use $C[u,v]$
to denote one of them
that is regarded as the arc of the chord $(u,v)$ in the context.
When necessary,
we will state, e.g., ``the arc $\cC[u,v]$''
for specifying which one is currently regarded as the arc.
A \newwd{gap-$d$} (resp., \newwd{gap-$d^+$}) \newwd{chord}
is a chord $(u,v)$ whose arc $C[u,v]$ is of length $d+2$
(resp., length $\ge d+2$).
For any chord $(u,v)$,
the subplane inside of the cycle $C$
surrounded by the chord $(u,v)$ and the arc $C[u,v]$
is called the \newwd{lower area} of the chord;
on the other hand,
the other side of the chord within the cycle $C$
is called the \newwd{upper area} (see Figure~\ref{fig:chord_arc}(a)).
A \newwd{lowest gap-$d^+$} chord
is a gap-$d^+$ chord that has no other gap-$d^+$ chord in its lower area.
We say that
two chords $(u_1,v_1)$ and $(u_2,v_2)$ \newwd{cross}
if they cross in the circle $C$ in a natural way
(see Figure~\ref{fig:chord_arc}(b)).
Formally,
we say that
$(u_1,v_1)$ {\em crosses} $(u_2,v_2)$
if either
(i) $u_2$ is on the tour $\cC(u_1,v_1)$ and
$v_2$ is on the tour $\aC(u_1,v_1)$,
or
(ii) $v_2$ is on the tour $\cC(u_1,v_1)$ and
$u_2$ is on the tour $\aC(u_1,v_1)$.
Also,
we say that $(u_1,v_1)$ \newwd{semi-crosses} $(u_2,v_2)$
if either
(i) $u_2$ is on the tour $\cC[u_1,v_1]$ and
$v_2$ is on the tour $\aC[u_1,v_1]$,
or
(ii) $v_2$ is on the tour $\cC[u_1,v_1]$ and
$u_2$ is on the tour $\aC[u_1,v_1]$
(see Figure~\ref{fig:chord_arc}(b)).
Note that clearly crossing implies semi-crossing.
In addition,
we say that
a chord $(u_1, v_1)$ \newwd{separates} two chords $(u_2, v_2)$ and $(u_3, v_3)$
if the endpoints of two chords $v_2$ and $v_3$ are separated by the chord $(u_1, v_1)$ (see Figure~\ref{fig:chord_arc}(c)).
Formally,
$(u_1, v_1)$ {\em separates} $(u_2, v_2)$ and $(u_3, v_3)$
if either
(i) $v_2$ is on the tour $\cC[u_1,v_1]$ and
$v_3$ is on the tour $\aC[u_1,v_1]$,
or
(ii) $v_3$ is on the tour $\cC[u_1,v_1]$ and
$v_2$ is on the tour $\aC[u_1,v_1]$.
We say that $k$ chords $(u_1, v_1)$, $(u_2, v_2),\ldots,(u_k, v_k)$ are \newwd{traversable}
if the following two conditions are satisfied:
\begin{enumerate}
	\item $(u_1, v_1)$ semi-crosses $(u_2, v_2)$,
	\item $\forall i \in [3, k]$, $\exists p, q < i$, $(u_i, v_i)$ separates $(u_p, v_p)$ and $(u_q, v_q)$.
\end{enumerate}

Now for the graph $\blcG=(\blcV,\blcE)$,
we define the circle graph $\cirblcG=(\cirblcV,\cirblcE)$ by

\[
\begin{array}{lcl}
\cirblcV&=&\rimV,{\rm~and}\\
\cirblcE
&=&
\bigr\{\,(u,v)\,|\,
 \mbox{$\exists$path from $u$ to $v$ in $\blcG$}\,\bigl\},
\end{array}
\]

where we assume that
the rim vertices of $\cirblcV$ ($=\rimV$) are placed on a cycle $\blcC$
as they are on the rim of the block in the grid graph.
Then it is clear that
$\cirblcG$ keeps
the same reachability relation among vertices in $\cirblcV=\rimV$.
Recall that
$\blcG$ has $O(n^2)$ vertices.
Thus,
by using {\sf PlanarReach},
we can show the following lemma.

\begin{lemma}\label{lem:trans_circle}
$\cirblcG$ keeps
the same reachability relation among vertices in $\cirblcV=\rimV$.
That is,
for any pair $u,v$ of vertices of $\cirblcV$,
$v$ is reachable from $u$ in $\cirblcG$
if and only if it is reachable from $u$ in $\blcG$.
There exists an algorithm
that transforms $\blcG$ to $\cirblcG$
in $O(n)$-space and polynomial-time in $n$.
\end{lemma}

The notion of traversable is a key for discussing the reachability on $\cirblcG$.
Based on the following lemma,
we use a traversable sequence of edges for characterizing the reachability on the circle graph $\cirblcG$.

\begin{figure}
\begin{center}
\includegraphics[width= 10 cm]{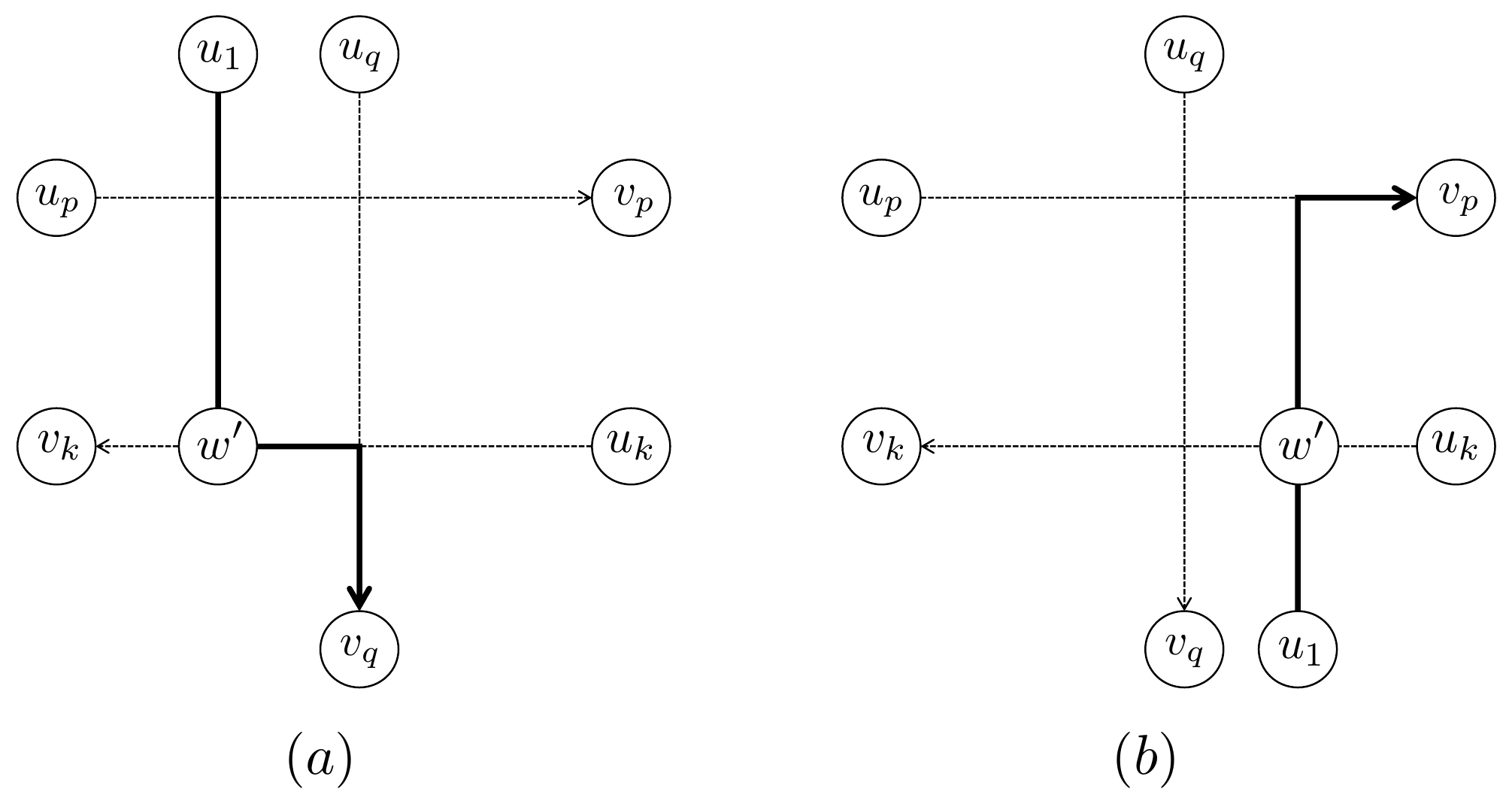}
\caption{A common vertex $w'$ of a path from $u_k$ to $v_k$ and a path from $u_1$ to $v_{q}$ or $v_{p}$ for some $p, q < k$.}
\label{fig:plane_cross}
\end{center}
\end{figure}

\begin{lemma}\label{lem:cross_reach}
	For a circle graph $\cirblcG = (\cirblcV, \cirblcE)$ obtained from a block grid graph $\blcG$, if there are traversable edges $(u_1, v_1)$, $(u_2, v_2),\ldots,(u_k, v_k) \in \cirblcE$, then $(u_1, v_k) \in \cirblcE$.
\end{lemma}

\begin{proof}
We show that $v_k$ is reachable from $u_1$ in $\blcG$ by induction on $k$. 
First, we consider the case $k=2$,
namely $(u_1, v_1)$ semi-crosses $(u_2,v_2)$.
$\blcG$ contains a path $p_{u_1,v_1}$ which goes from $u_1$ to $v_1$.
Also, $\blcG$ contains a path $p_{u_2,v_2}$ which goes from $u_2$ to $v_2$.
Since $\blcG$ is planar and $u_1$, $v_1$, $u_2$, and $v_2$ are the rim vertices and the edges are semi-crossing,
there exists a vertex $w$ which is common in $p_{u_1,v_1}$ and $p_{u_2, v_2}$ in $\blcG$.
Since $w$ is reachable from $u_1$ and $v_2$ is reachable from $w$,
there exists a path from $u_1$ to $v_2$.

Next, we assume that the lemma is true for all sequences of traversable edges of length less than $k$.
By the definition,
there exist two edges $(u_p, v_p)$ and $(u_q, v_q)$ that the edge $(u_k, v_k)$ separates ($p, q < k$).
We have two paths $p_{u_1, v_{p}}$ from $u_1$ to $v_{p}$ and $p_{u_1, v_{q}}$ from $u_1$ to $v_{q}$ in $\blcG$ by the induction hypothesis.
Also we have a path $p_{u_k, v_k}$ from $u_k$ to $v_k$.
Since $(u_k, v_k)$ separates $(u_p, v_p)$ and $(u_q, v_q)$,
$v_{p}$ and $v_{q}$ are on the different sides of arcs of the edge $(u_k, v_k)$.
If $u_1$ and $v_{p}$ are on the same arc of $(u_k, v_k)$,
the paths $p_{u_1, v_{q}}$ and $p_{u_k, v_k}$ have a common vertex $w'$ (see Figure~\ref{fig:plane_cross}(a)).
On the other hand,
if $u_1$ and $v_{q}$ are on the same arc of $(u_k, v_k)$,
the paths $p_{u_1, v_{p}}$ and $p_{u_k, v_k}$ have a common vertex $w'$ (see Figure~\ref{fig:plane_cross}(b)).
Thus there exists a path from $u_1$ to $v_k$ via $w'$ in $\blcG$.
\end{proof}

\subsection{Gadget Graph}

We introduce the notion of ``gadget graph''.
A gadget graph is a graph that is given a ``label set'' to each edge.

\begin{figure}
\begin{center}
\includegraphics[width= 10 cm]{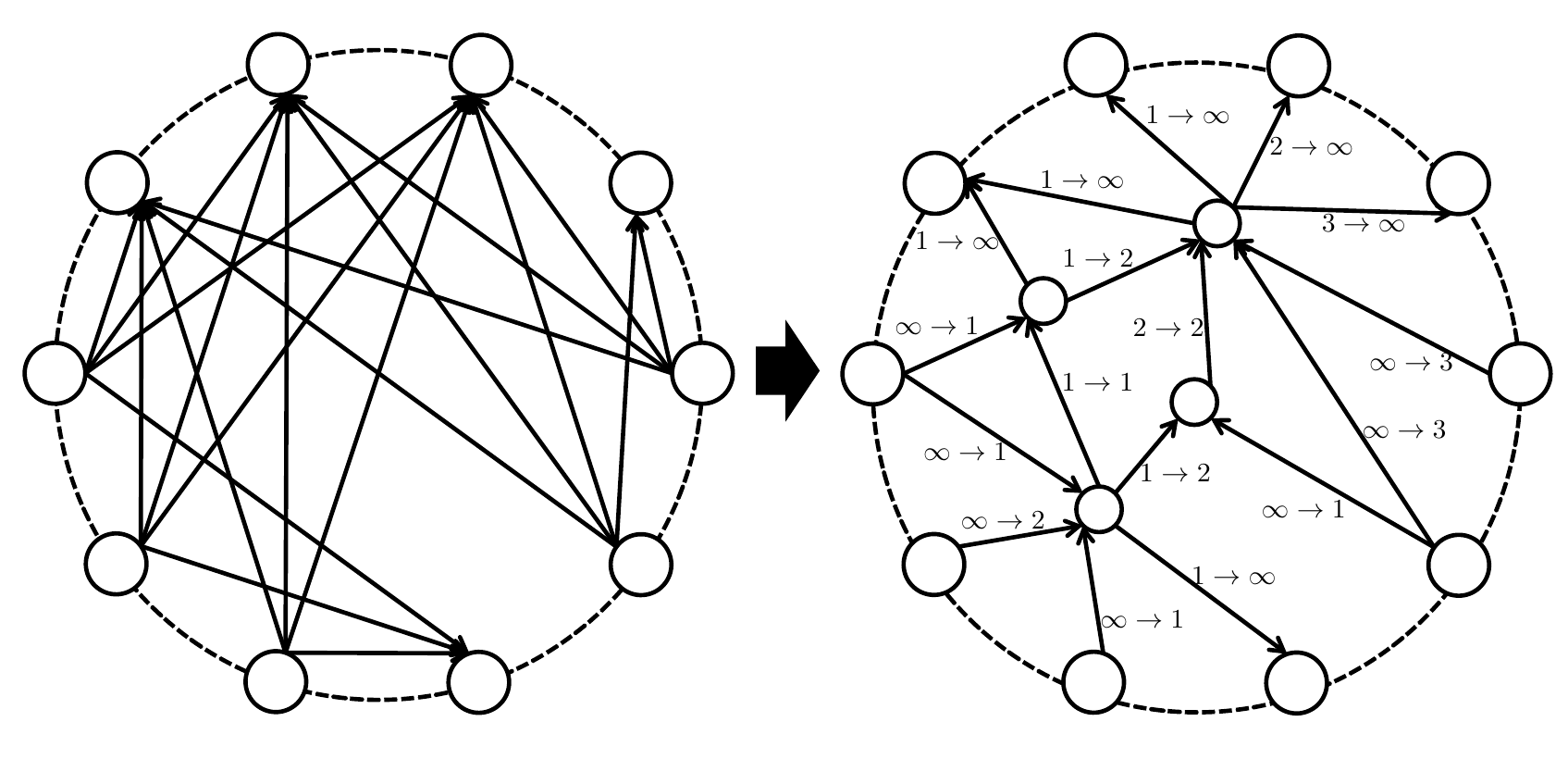}
\caption{An example of the transformation
from a circle graph to a gadget graph.}
\label{fig:trans_gadget}
\end{center}
\end{figure}

\begin{definition}\label{def:gadget}
A \newwd{gadget graph} $\gadG$ is a graph
defined by a tuple $(\gadV,\gadE,\gadK,\gadL)$,
where $\gadV$ is a set of vertices,
$\gadE$ is a set of edges,
$\gadK$ is a \newwd{path function} that assigns an edge or $\bot$ to each edge,
and $\gadL$ is
a \newwd{level function} that assigns a \newwd{label set} to each edge.
A label set is a set $\{i_1\rightarrow o_1, i_2\rightarrow o_2,\ldots,i_k \rightarrow o_k\}$ of labels
where each label $i_j\rightarrow o_j$,
$i_j,o_j\in\mathbb{R}\cup\{\infty\}$,
is a pair of \newwd{in-level} and \newwd{out-level}.

\end{definition}

\noindent{\bf Remark.} For an edge $(u, v)\in \gadE$,
we may use expressions $\gadK(u, v)$ and $\gadL(u, v)$ instead of $\gadK((u, v))$ and $\gadL((u, v))$ for simplicity.

\ 

Our goal is to transform
a given circle graph
(obtained from a block grid graph) $\cirblcG=(\cirblcV,\cirblcE)$ in which all vertices in $\cirblcV$ are placed on a cycle $C$
to a {\em plane} gadget graph
$\gadplG=(\gadplVout\cup\gadplVin,\gadplE,\gadplK,\gadplL)$
where $\gadplVout$ is the set of \newwd{outer vertices}
that are exactly the vertices of $\cirblcV$ placed in the same way as $\cirblcG$ on the cycle $C$,
and $\gadplVin$ is the set of \newwd{inner vertices} placed inside of $C$.
All edges of $\gadplE$ are also placed inside of $C$ under our embedding.
The inner vertices of $\gadplVin$ are used
to replace crossing points of edges of $\cirblcE$
to transform to a planar graph
(see Figure~\ref{fig:trans_gadget}).
We would like to keep the ``reachability'' among vertices
in $\gadplVout$ in $\gadplG$
while bounding $|\gadplVin|=O(n)$.

We explain how to characterize the reachability on a gadget graph.
Consider any gadget graph $\gadG=(\gadV,\gadE,\gadK,\gadL)$,
and let $x$ and $y$ be any two vertices of $\gadV$.
Intuitively,
the reachability from $x$ to $y$ is characterized by a directed path
on which we can send a token from $x$ to $y$.
Suppose that
there is a directed path $p=(e_1,\ldots,e_m)$ from $x$ to $y$.
We send a token through this path.
The token has a level, which is initially $\infty$ when the token is at vertex $x$.
(For a general discussion,
we use a parameter $\ell_s$ for the initial level of the token.)
When the token reaches the tail vertex $t(e_j)$ of some edge $e_j$ of $p$ with level $\ell$,
it can ``go through'' $e_j$ to reach its head vertex $h(e_j)$
if $\gadL(e_j)$ has an \newwd{available label} $i_j\rightarrow o_j$ such that $i_j \le \ell$ holds for its in-level $i_j$.
If the token uses a label $i_j\rightarrow o_j$,
then its level becomes the out-level $o_j$ at the vertex $h(e_j)$.
If there are several available labels,
then we naturally use the one with the highest out-level.
If the token can reach $y$ in this way,
we consider that a ``token tour'' from $x$ to $y$ is ``realized'' by this path $p$.
Technically,
we introduce $\gadK$ so that some edge can specify the next edge.
We consider only a path $p = (e_1, \ldots, e_m)$ as ``valid'' such that $e_{i+1} = \gadK(e_i)$
for all $e_i$ such that $\gadK(e_i)\neq \bot$.
We characterize the reachability from $x$ to $y$ on gadget graph $\gadG$
by using a valid path realizing a token tour from $x$ to $y$.

\begin{definition}\label{def:reachable}
For any gadget graph $\gadG=(\gadV,\gadE,\gadK,\gadL)$,
and for any two vertices $x,y$ of $\gadV$,
there exists a \newwd{token tour} from $x$ to $y$ with initial level $\ell_s$
if there exists a sequence of edges $(e_1,\ldots,e_m)$ that satisfies

\begin{enumerate}
\item
$x=t(e_1)$ and $y=h(e_m)$,
\item
$h(e_i) = t(e_{i+1})$ $(1\le i < m)$,
\item
if $\gadK(e_i)$ is not $\bot$ $(1\le i < m)$,
then $e_{i+1} = \gadK(e_i)$,
\item
there exist labels $i_1\rightarrow o_1\in \gadL(e_1),\ldots,i_m\rightarrow o_m\in \gadL(e_m)$
such that $\ell_s \ge i_1$ and $o_t\ge i_{t+1}$ for all $1\le t<m$.
\end{enumerate}
\end{definition}

\begin{figure}[t]
\begin{center}
\includegraphics[width= 11 cm]{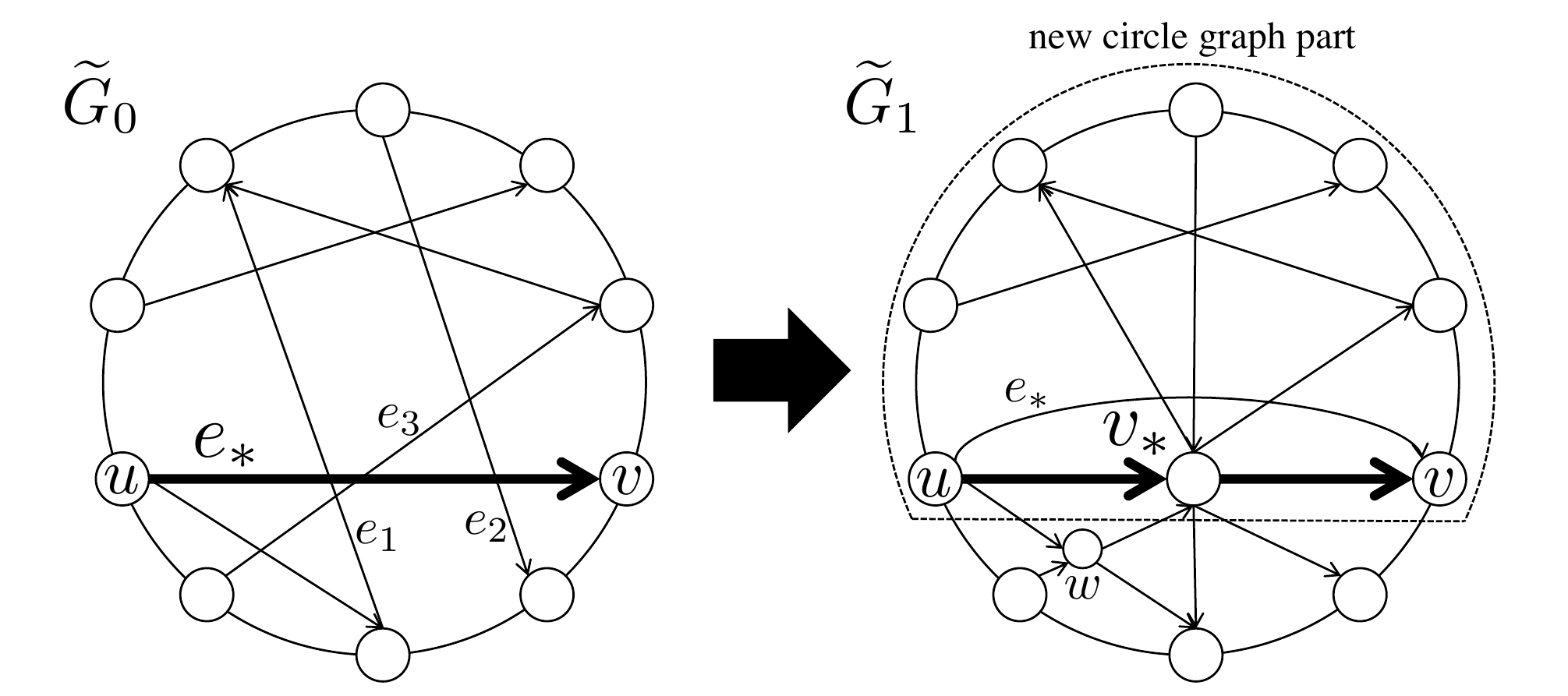}
\caption{An initial transformation step from $\gadG_0$ to $\gadG_1$.}
\label{fig:process0to1}
\end{center}
\end{figure}

At the beginning of our algorithm,
we obtain a gadget graph $\gadG_0 = (\gadV_0, \gadE_0, \gadK_0, \gadL_0)$ whose base graph is equal to $\cirblcG$,
and $\gadK_0(e) = \bot$, $\gadL_0(e) = \{0\rightarrow \infty\}$ for every $e\in \gadE_0$.
It is obvious that $\cirblcG$ and $\gadG_0$ have the same reachability.
Namely,
there exists a token tour from $x$ to $y$ for $x, y\in \gadV_0$ in $\gadG_0$ if and only if
there exists an edge $(x, y) \in \gadE_0$.

We explain first the outline of our transformation
from $\gadG_0$ to $\gadplG$.
We begin by finding a chord $e_*=(u,v)$ with gap $\ge 2$
having no other gap-$2^+$ chord in its lower area,
that is,
one of the lowest gap-$2^+$ chords.
(If there is no gap-$2^+$ chord, then the transformation is terminated.)
For this $e_*$ and its lower area,
we transform them into a planar part and reduce the number of crossing points as follows
(see Figure~\ref{fig:process0to1}):
(i) Consider all edges of $\gadG_0$ crossing this chord $e_*$ ($e_1$, $e_2$ and $e_3$ in Figure~\ref{fig:process0to1}).
Create a new inner vertex $v_*$ of $\gadplG$ on the chord
and bundle all crossing edges going through this vertex $v_*$; 
that is,
we replace all edges crossing $e_*$
by edges between their end points in the lower area of $e_*$ and $v_*$,
and edges between $v_*$ and their end points in the upper area of $e_*$.
(ii) Introduce new inner vertices
for edges crossing gap-$1$ chords in the lower area of $e_*$ ($w$ in Figure~\ref{fig:process0to1}).
(iii) Add appropriate label sets to those newly introduced edges
so that the reachability is not changed by this transformation.
At this point
we regard the lower area of $e_*$ as processed,
and remove this part from the circle graph part of $\gadG_0$
by replacing the arc $C[u,v]$ by a tour $(u,v_*,v)$
to create a new circle graph part of $\gadG_1$.
We then repeat this transformation step on the circle graph part of $\gadG_1$.
In the algorithm,
$U_t$ is the vertices of the circle graph part of $\gadG_t$,
thus $\gadG_t[U_t]$ indicates the circle graph part of $\gadG_t$.
Note that $e_*$ is not removed and becomes a gap-1 chord in the next step.

We explain step (ii) for $\gadG_0$ in more detail.
Since $e_*$ is a gap-$2^+$ chord,
there exist only gap-$1$ chords or edges whose one end point is $v_*$
in the lower area of $e_*$.
If there are two edges $e_0$ and $e_1$ that cross each other,
we replace the crossing point by a new inner vertex $u$ (see Figure~\ref{fig:make_planar}(a), (b)).
The edge $e_i$ becomes two edges $(t(e_i), u)$ and $(u, h(e_i))$ $(i = 0, 1)$,
and we set $\gadK_1(t(e_i), u) = (u, h(e_i))$.
The edges might be divided into more than two segments (see Figure~\ref{fig:make_planar}(c)).
We call the edge of $\gadG_0$ \newwd{original edge} of the divided edges.
By the path function,
we must move along the original edge.
An edge $e$ might have a reverse direction edge $\bar{e} = (h(e), t(e))$ (see Figure~\ref{fig:make_planar}(d)).
In this case,
$e$ and $\bar{e}$ share a new vertex for resolving crossing points.
For $\gadG_t[U_t]$ $(t > 0)$,
we process the lower area in the same way.
We refer to this algorithm as {\sf MakePlanar},
and the new inner vertices created by {\sf MakePlanar} in step $t$ as $V_{\sf MP}^t$.

\begin{figure}[bt]
\begin{center}
\includegraphics[width= 14 cm]{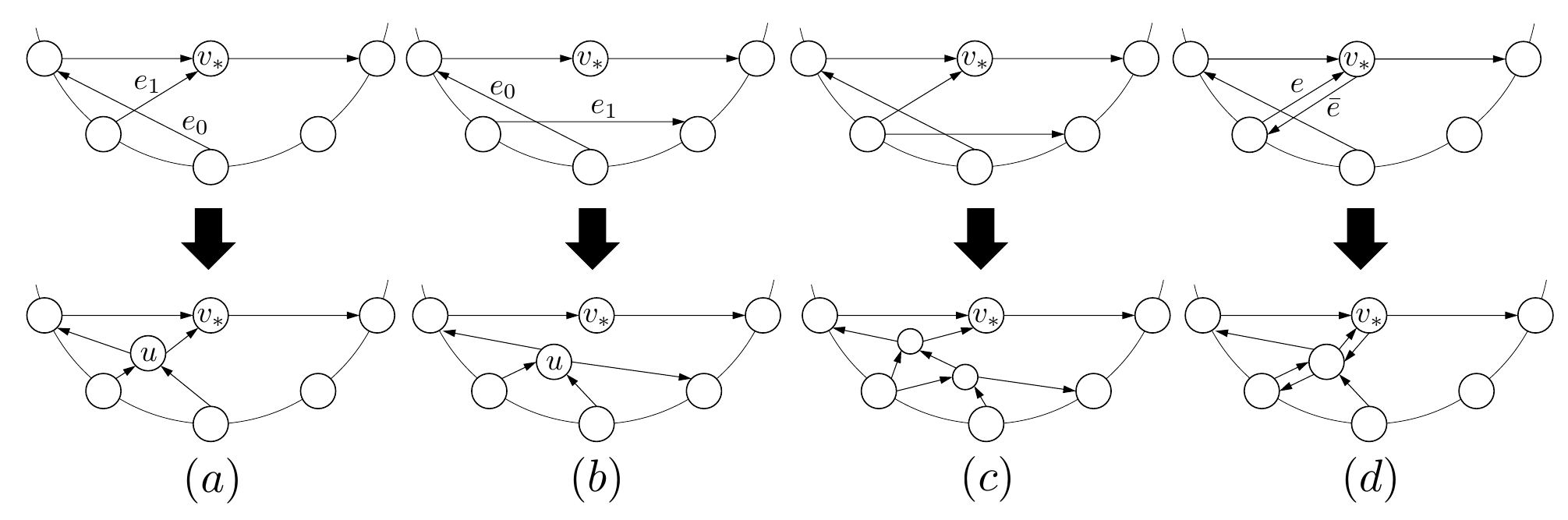}
\caption{Examples of vertices made by {\sf MakePlanar}.}
\label{fig:make_planar}
\end{center}
\end{figure}

\begin{algorithm}
\caption{}\label{alg:trans_gadget}
\begin{algorithmic}[1]
\REQUIRE A circle graph $\cirblcG=(\cirblcV, \cirblcE)$ obtained from a block graph.
\ENSURE Output a plane gadget graph $\gadplG = (\gadplVout \cup \gadplVin, \gadplE, \gadplK, \gadplL)$ which satisfies $\gadplVout = \cirblcV$ and the reachability among vertices in $\gadplVout$ in $\gadplG$ is the same as $\cirblcG$.
\STATE initialize $t = 0$ // loop counter
\STATE $\gadG_0 = (\gadVout \cup \gadV_0, \gadE_0, \gadK_0, \gadL_0)$ where $\gadVout \leftarrow \cirblcV, \gadV_0 \leftarrow \emptyset, \gadE_0 \leftarrow \cirblcE, \gadK_0(e) \leftarrow \bot, \gadL_0(e) \leftarrow \{0\rightarrow \infty\}$ for each $e \in \cirblcE$, and $U_0 \leftarrow \gadVout$
\STATE for every $v \in \gadVout$, $\lin^0(v)\leftarrow 0,\lout^0(v)\leftarrow \infty, p^0(v) \leftarrow v$.
\WHILE {$\gadG_{t}[U_t]$ has a lowest gap-$2^+$ chord}
\STATE pick a lowest gap-$2^+$ chord $e_*^t$
\STATE make a new vertex $v_*^t$
\STATE $\gadV_{t+1}\leftarrow \gadV_{t}\cup \{v_*^t\}$
\STATE $\gadE_{t+1}\leftarrow (\gadE_{t}\cup \{(t(e), v_*^t), (v_*^t, h(e))\ |\ e$ crosses $e_*$ or $e = e_*\}) \setminus \{e\ |\ e$ crosses $e_*\}$
\STATE $U_{t+1}\leftarrow (U_t \cup \{v_*^t\}) \setminus C_{\gadG_{t}[U_t]}(t(e_*^t), h(e_*^t))$
\STATE use {\sf MakePlanar} to make the lower area of $e_*^t$ planar and update $\gadV_{t+1}$, $\gadE_{t+1}$ and $\gadK_{t+1}$.
\STATE change the labels by using Algorithm~\ref{alg:change_label} for keeping reachability
\STATE output $\gadG_{t+1}[C_{\gadG_{t}[U_t]}[t(e_*^t), h(e_*^t)]\cup \{v_*^t\} \cup V_{\sf MP}^t]$, which is the lower area of $e_*^t$.
\STATE $t\leftarrow t+1$
\ENDWHILE
\STATE use {\sf MakePlanar} to make $\gadG_t[U_t]$ planar and assign labels by line 17-24 of Algorithm~\ref{alg:change_label}.
\STATE output $\gadG_{t}[U_t \cup V_{\sf MP}^t]$
\end{algorithmic}
\end{algorithm}

\begin{algorithm}
\caption{}\label{alg:change_label}
\begin{algorithmic}[1]
\ENSURE Set $\gadL_{t+1}$ so that $\gadG_{t+1}$ has the same reachability as $\gadG_t$
\STATE For every edge $e$ appearing in both $\gadG_t$ and $\gadG_{t+1}$, let $\gadL_{t+1}(e) = \gadL_{t}(e)$.
\STATE $S^\ell$ (resp., $S^u$) $\leftarrow \{v\in U_t\ |\ \exists e \in \gadE_t$ s.t. $e$ crosses $e_*^t$, $t(e) = v$ or $h(e) = v$, and $v$ is at the lower (resp., upper) area of $e_*^t\}$
\STATE $T^\ell$ (resp., $T^u$) $\leftarrow \{v\in \cirblcV\ |\ p^t(v) \in S^\ell$ (resp. $S^u$)$\}$
\STATE Fix any vertices $x', y' \in \cirblcV$ such that $p^t(x') = t(e_*^t), p^t(y') = h(e_*^t)$.
\STATE Set an order to $T^\ell$ according to the order appearing in $C_{\cirblcG}[y', x']$. We regard $T^\ell$ as a sequence $(t_1^\ell, t_2^\ell, \ldots, t_{|T^\ell|}^\ell)$ (see Figure~\ref{fig:ST_ul}(b)).
\STATE Set an order to $T^u$ in the same way as $T^\ell$ but according to the tour along the other arc. We also regard $T^u$ as a sequence $(t_1^u, t_2^u, \ldots, t_{|T^u|}^u)$ (see Figure~\ref{fig:ST_ul}(b)). 
\STATE Use Algorithm~\ref{alg:calc_level} for calculating $\lin^{t+1}(v)$ and $\lout^{t+1}(v)$ for all $v \in T^\ell$.

\FOR {$u \in S^\ell$}
\STATE $\gadL_{t+1}(u, v_*^t) \leftarrow \{\lin^t(v) \rightarrow \lin^{t+1}(v)\ |\ p^t(v) = u\}$
\STATE $\gadL_{t+1}(v_*^t, u) \leftarrow \{\lout^{t+1}(v) \rightarrow \lout^t(v)\ |\ p^t(v) = u\}$
\ENDFOR

\FOR {$u \in S^u$}
\STATE $\gadL_{t+1}(u, v_*^t) \leftarrow \{\lin^{t}(t_i^u) \rightarrow \max_{t^\ell \in T^\ell}\{\lout^{t+1}(t^\ell)|\ (t_i^u, t^\ell) \in \cirblcE\}|t_i^u \in T^u$ and $p^t(t_i^u) = u\}$
\STATE $\gadL_{t+1}(v_*^t, u) \leftarrow \{\min_{t^\ell \in T^\ell}\{\lin^{t+1}(t^\ell)\ |\ (t^\ell, t_i^u) \in \cirblcE\} \rightarrow \lout^{t}(t_i^u)|\ t_i^u \in T^u$ and $p^t(t_i^u) = u\}$
\ENDFOR

\STATE $\gadL_{t+1}(t(e_*^t), v_*^t) \leftarrow \{\infty \rightarrow 0\}$, $\gadL_{t+1}(v_*^t, h(e_*^t)) \leftarrow \{\infty \rightarrow 0\}$

\FORALL {edge $e$ created by {\sf MakePlanar}}
\STATE Let $e'$ be the original edge of $e$
\IF {$t(e) = t(e')$}
\STATE $\gadL_{t+1}(e) \leftarrow \{a \rightarrow b\ |\ a\rightarrow b \in \gadL_t(e')\}$
\ELSE
\STATE $\gadL_{t+1}(e) \leftarrow \{b \rightarrow b\ |\ a\rightarrow b \in \gadL_t(e')\}$
\ENDIF
\ENDFOR

\FOR{$v\in \{u\in\cirblcV\ |\ \exists w\in U_t$ s.t. $w$ is at the lower area of $e_*^t$ and $p^t(u) = w\}$}
\STATE $p^{t+1}(v) = v_*^t$
\ENDFOR

\STATE Unchanged $\lin^t(\cdot), \lout^t(\cdot)$ and $p^t(\cdot)$ will be taken over to $\lin^{t+1}(\cdot), \lout^{t+1}(\cdot)$ and $p^{t+1}(\cdot)$.
\end{algorithmic}
\end{algorithm}

\begin{algorithm}
\caption{}\label{alg:calc_level}
\begin{algorithmic}[1]
\ENSURE Calculate $\lin^{t+1}(v)$ and $\lout^{t+1}(v)$ for all $v \in T^\ell$.
\FOR {$i \in [1,|T^\ell|]$}
\STATE $\lin^{t+1}(t_i^\ell) \leftarrow \max \{j\ |\ (t_i^\ell, t_j^u) \in \cirblcE,\ t_j^u\in T^u\} + i / n$
\STATE $\lout^{t+1}(t_i^\ell) \leftarrow \min \{j\ |\ (t_j^u, t_i^\ell) \in \cirblcE,\ t_j^u\in T^u\} + i / n$
\ENDFOR
\FOR {$i = 1$ to $|T^\ell|$}
\STATE $\Delta \leftarrow \max(0, \max\{\lout^{t+1}(t_j^\ell) - j/n\ |\ 1\le j < i\} - (\lin^{t+1}(t_i^\ell) - i / n))$
\FOR {$k \in [i, |T^\ell|]$}
\STATE $\lin^{t+1}(t_k^\ell) \leftarrow \lin^{t+1}(t_k^\ell) + \Delta$
\STATE $\lout^{t+1}(t_k^\ell) \leftarrow \lout^{t+1}(t_k^\ell) + \Delta$
\ENDFOR
\ENDFOR

\end{algorithmic}
\end{algorithm}

The detailed process of step (iii) is written in Algorithm~\ref{alg:change_label},
and Algorithm~\ref{alg:trans_gadget} describes the entire process of step (i), (ii) and (iii).
The following lemma shows that an output graph of Algorithm~\ref{alg:trans_gadget} has small size.

\begin{lemma}\label{lem:trans_gadget_small}
Algorithm~\ref{alg:trans_gadget} terminates
creating a planar graph of size $O(n)$.
\end{lemma}

\begin{proof}
In the beginning of the algorithm,
$|U_0| = n$ and $|U_t|$ decreases by at least 1 for each iteration
since the picked edge $e_*^t$ is a gap-$2^+$ chord.
Hence the algorithm stops after at most $n$ iterations
and the number of the new inner vertices made at line 7, or $v_*^t$, is also at most $n$.
If a gap-$k$ chord is picked,
we make at most $2k-1$ new inner vertices by {\sf MakePlanar},
namely $|V_{\sf MP}^t| \le 2k-1$,
since there exist only gap-$1$ chords in the lower area of the picked edge.
The total number of inner vertices becomes at most
\[
	n + \sum_{i=1}^t (2k_i-1) = n + 2\sum_{i=1}^t k_i - t \le n + 2\times 2n = 5n
\]
where $t$ is the number of iterations and $k_i$ means that a gap-$k_i$ chord was picked in the $i$-th iteration. After all, $|\gadplVout \cup \gadplVin|\le n + 5n = 6n$.
\end{proof}

Now we explain Algorithm~\ref{alg:change_label} describing how to assign labels to $\gadG_{t+1}$ constructed in Algorithm~\ref{alg:trans_gadget}.
For each outer vertex $v \in \gadVout$,
we keep three attributes $p^t(v)$, $\lin^t(v)$ and $\lout^t(v)$,
and we call them \newwd{parent}, \newwd{in-level} and \newwd{out-level} respectively.
We calculate these values from line 2 to 7 and line 25 to 27.
$p^t(v)$ is a vertex belonging to the circle graph part of $\gadG_t$, namely $p^t(v) \in U_t$.
From the algorithm,
we can show that there are token tours from $v$ to $p^t(v)$ and/or from $p^t(v)$ to $v$.
For the token tour from $v$ to $p^t(v)$,
the final level of the token becomes $\lin^t(v)$.
On the other hand,
for the token tour from $p^t(v)$,
it is enough to have $\lout^t(v)$ as an initial level to reach $v$.
We will show these facts implicitly in the proof of Lemma~\ref{lem:main_lem1}.

At the beginning of each iteration of Algorithm~\ref{alg:trans_gadget},
we choose a lowest gap-$2^+$ chord $e_*^t$.
We collect vertices in $U_t$ which are endpoints for some edges crossing with $e_*^t$,
and we refer to the vertices among them which are in the lower area of $e_*^t$ as $S^\ell$
and the vertices in the upper area of $e_*^t$ as $S^u$ (see Figure~\ref{fig:ST_ul}(a) and line 2).
Next we collect vertices whose parents are in $S^\ell$ (resp., $S^u$),
and we denote them by $T^\ell$ (resp., $T^u$) (line 3).
Let $x'$ and $y'$ be vertices 
whose parents are $t(e_*^t)$ and $h(e_*^t)$ respectively.
We assign indices to the vertices in $T^u$ and $T^\ell$
such that the nearer to $x'$ a vertex is located, the larger index the vertex has (see Figure~\ref{fig:ST_ul}(b)).
We regard $T^\ell$ as a sequence $(t_1^\ell, t_2^\ell, \ldots, t_{|T^\ell|}^\ell)$,
and $T^u$ as a sequence $(t_1^u, t_2^u, \ldots, t_{|T^u|}^u)$.
For each vertex $t_i^\ell$ in $T^\ell$,
we calculate $\lin^{t+1}(t_i^\ell)$ and $\lout^{t+1}(t_i^\ell)$ in Algorithm~\ref{alg:calc_level}.
From line 1 to 4,
we decide temporary values of $\lin^{t+1}(t_i^\ell)$ and $\lout^{t+1}(t_i^\ell)$ according to reachability among vertices in $T^\ell$ and $T^u$ in $\cirblcG$.
When $t_j^u$ has the maximum index among vertices that $t_i^\ell$ can reach in $T^u$,
we let $\lin^{t+1}(t_i^\ell) = j + i / n$.
When $t_j^u$  has the minimum index among vertices which can reach $t_i^\ell$ in $T^u$,
we let $\lout^{t+1}(t_i^\ell) = j + i / n$.
The term $i/n$ is for breaking ties.
See Figure~\ref{fig:ell_in_out}:
$T^\ell = \{t_1^\ell, t_2^\ell, t_3^\ell\}$, $T^u = \{t_1^u, t_2^u, t_3^u\}$
and the edges are derived from $\cirblcE$.
The vertex $t_3^\ell$ can reach $t_1^u$, $t_2^u$ and $t_3^u$.
Thus $\lin^{t+1}(t_3^\ell) = \max(1, 2, 3)+3/n = 3+3/n$.
The vertices $t_2^u$ and $t_3^u$ can reach $t_2^\ell$.
Thus $\lout^{t+1}(t_2^\ell) = \min(2, 3) + 2/n = 2 + 2/n$.
In the next for-loop, we change the in- and out-levels
so that the in-level of the larger indexed vertex is larger than the out-level of the smaller indexed vertex.
If there exists a vertex $t_j^\ell$ such that $i > j$ and $\lout^{t+1}(t_j^\ell) > \lin^{t+1}(t_i^\ell)$,
then we let $\Delta = (\lout^{t+1}(t_j^\ell) - j/n) - (\lin^{t+1}(t_i^\ell) - i / n)$ and
add $\Delta$ to $\lin^{t+1}(t_i^\ell)$ and $\lout^{t+1}(t_i^\ell)$.
For preserving the magnitude relationship between in- and out-levels of $t_i^\ell$ and
those of $t_k^\ell$ ($k > i$),
we also add $\Delta$ to $\lin^{t+1}(t_k^\ell)$ and $\lout^{t+1}(t_k^\ell)$.
In Figure~\ref{fig:ell_in_out},
we have $\lin^{t+1}(t_2^\ell) < \lout^{t+1}(t_1^\ell)$.
Thus
we add $1 = (\lout^{t+1}(t_1^\ell)-1/n) - (\lin^{t+1}(t_2^\ell) - 2/n)$ to $\lin^{t+1}(t_2^\ell)$.
Moreover,
we add 1 to $\lout^{t+1}(t_2^\ell)$, $\lin^{t+1}(t_3^\ell)$ and $\lout^{t+1}(t_3^\ell)$
so that
we keep the magnitude relationship.
Here we state a lemma.

\begin{lemma}\label{lem:in_out}
  For any $t$,
  if $i > j$,
  then $\lin^{t+1}(t_i^\ell) > \lout^{t+1}(t_j^\ell)$.
\end{lemma}

Back to Algorithm~\ref{alg:change_label}.
From line 8 to 16,
we assign labels to edges newly appearing in $\gadG_{t+1}$.
Figure~\ref{fig:label} is an example of how to assign label sets based on Figure~\ref{fig:ell_in_out}.
The vertex $a$ is the parent of $t_3^u$,
$b$ is the parent of $t_1^u$ and $t_2^u$,
$c$ is the parent of $t_2^\ell$ and $t_3^\ell$ and
$d$ is the parent of $t_1^\ell$.
Let $v$ be any vertex in $T^\ell$.
For edges in the lower area of $e_*^t$,
the edge $(p^t(v), v_*^t)$ has a label $\lin^t(v) \rightarrow \lin^{t+1}(v)$ (line 9),
and the edge $(v_*^t, p^t(v))$ has a label $\lout^{t+1}(v) \rightarrow \lout^t(v)$ (line 10).
In Figure~\ref{fig:label},
the edge $(c, v_*^t)$ has labels $\lin^t(t_2^\ell) \rightarrow \lin^{t+1}(t_2^\ell)$ and $\lin^t(t_3^\ell) \rightarrow \lin^{t+1}(t_3^\ell)$.
The edge $(v_*^t, c)$ has a label $\lout^{t+1}(t_2^\ell) \rightarrow \lout^t(t_2^\ell)$ and
the edge $(v_*^t, d)$ has a label $\lout^{t+1}(t_1^\ell) \rightarrow \lout^t(t_1^\ell)$.
Consider edges in the upper area of $e_*^t$.
Let $v$ be any vertex in $T^u$.
The edge $(p^t(v), v_*^t)$ has a label $\lin^t(v) \rightarrow \ell_{max}$
where $\ell_{max}$ is the maximum in-level of vertices in $T^\ell$ that can reach $v$ (line 13).
The edge $(v_*^t, p^t(v))$ has a label $\ell_{min} \rightarrow \lout^t(v)$
where $\ell_{min}$ is the minimum out-level of vertices in $T^\ell$ that $v$ can reach (line 14).
In Figure~\ref{fig:label},
the edge $(a, v_*^t)$ has a label $\lin^t(t_3^u) \rightarrow \lout^{t+1}(t_2^\ell)$
since
$t_3^u$ can reach $t_1^\ell$ and $t_2^\ell$, and $\lout^{t+1}(t_1^\ell) < \lout^{t+1}(t_2^\ell)$ (see Figure~\ref{fig:ell_in_out}).
The edge $(v_*^t, b)$ has a label $\lin^{t+1}(t_2^\ell) \rightarrow \lout^{t}(t_1^u)$
since
$t_2^\ell$ and $t_3^\ell$ can reach $t_1^u$, and $\lin^{t+1}(t_2^\ell) < \lin^{t+1}(t_3^\ell)$ (see Figure~\ref{fig:ell_in_out}).
The edges $(t(e_*^t), v_*^t)$ and $(v_*^t, h(e_*^t))$ have only one label $\infty \rightarrow 0$, which prohibits using these edges (line 16).

From line 17 to 24,
we assign labels to edges made by {\sf MakePlanar}.
For every edge $(u, v)$ in the lower area of $e_*^t$,
the edge $(u, v)$ might be divided into some edges, for instance $(u, w_1), (w_1, w_2),\ldots (w_k, v)$ by {\sf MakePlanar}.
In this case,
when $(u, v)$ has a label $a \rightarrow b$,
$(u, w_1)$ has a label $a \rightarrow b$
and the other edges have labels $b \rightarrow b$ (see Figure~\ref{fig:mp_label}).

From line 25 to 27,
we update the parents of the vertices whose parents are in the lower area of $e_*^t$.
For each vertex $v$ in $\cirblcV$ which $p^t(v)$ is in the lower area of $e_*^t$,
we let $p^{t+1}(v) = v_*^t$.

\begin{figure}[bt]
\begin{center}
\includegraphics[width= 13 cm]{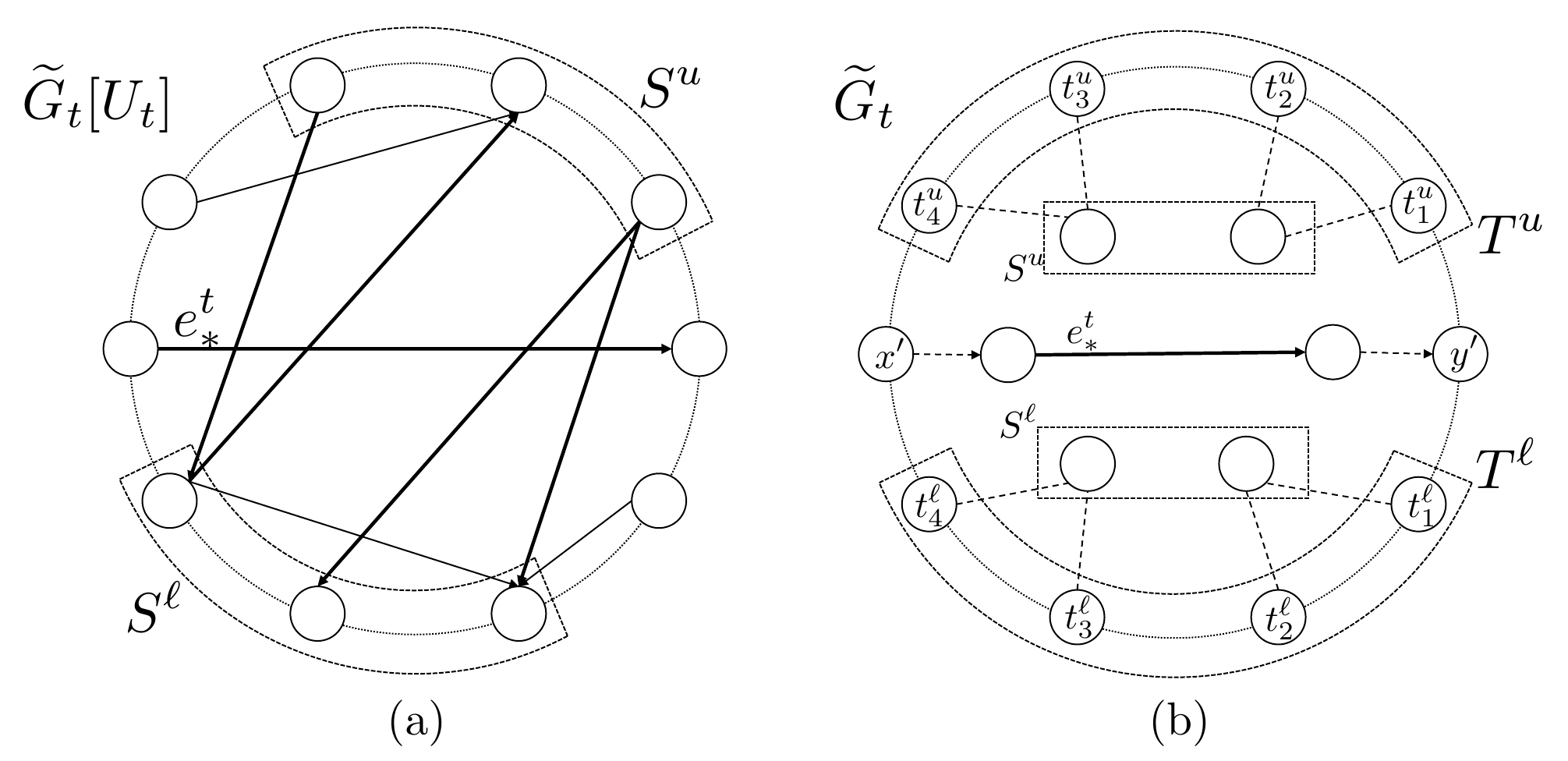}
\caption{(a) An example of $S^\ell$ and $S^u$, (b) An example of $T^\ell$ and $T^u$.}
\label{fig:ST_ul}
\end{center}
\end{figure}

For a gadget graph $\gadG = (\gadV, \gadE, \gadK, \gadL)$,
we use
$(v_1, \ell_1)\Rightarrow (v_2, \ell_2), \Rightarrow \cdots \Rightarrow (v_m, \ell_m)$ 
to denote a token tour from $v_1$ to $v_m$ in $\gadG$
with having a level $\ell_i$ at $v_i \in \gadV$ for any $1\le i \le m$. 
Needless to say,
if such a tour exists,
$(v_i, v_{i+1})\in \gadE$ and $\ell'_i \rightarrow \ell_{i+1} \in \gadL(v_i, v_{i+1})$ where $\ell_i \ge \ell'_i$ for any $1\le i < m$.
In addition,
when we would like to show which available labels we used, we write, for example, $(v_i, \ell_i; \ell'_i\rightarrow \ell_{i+1}) \Rightarrow (v_{i+1}, \ell_{i+1})$, which means the available label $\ell'_i \rightarrow \ell_{i+1}$ was used.
The following lemma shows that
paths in $\gadG_0$ remain in $\gadG_t$ for every $t$.

\begin{figure}
\begin{center}
\includegraphics[width= 8 cm]{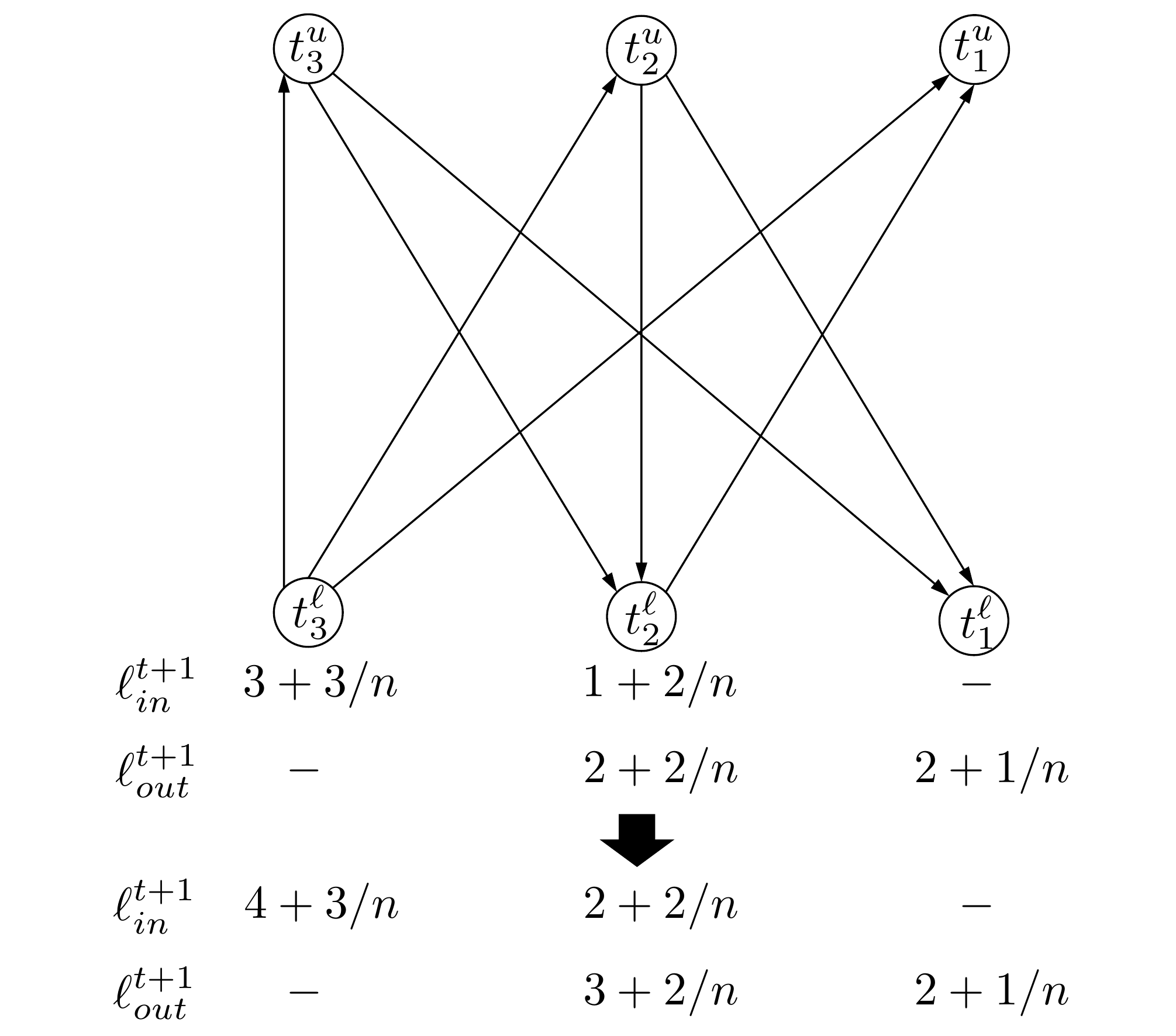}
\caption{How to calculate in and out levels.}
\label{fig:ell_in_out}
\end{center}
\end{figure}

\begin{figure}
\begin{center}
\includegraphics[width= 12 cm]{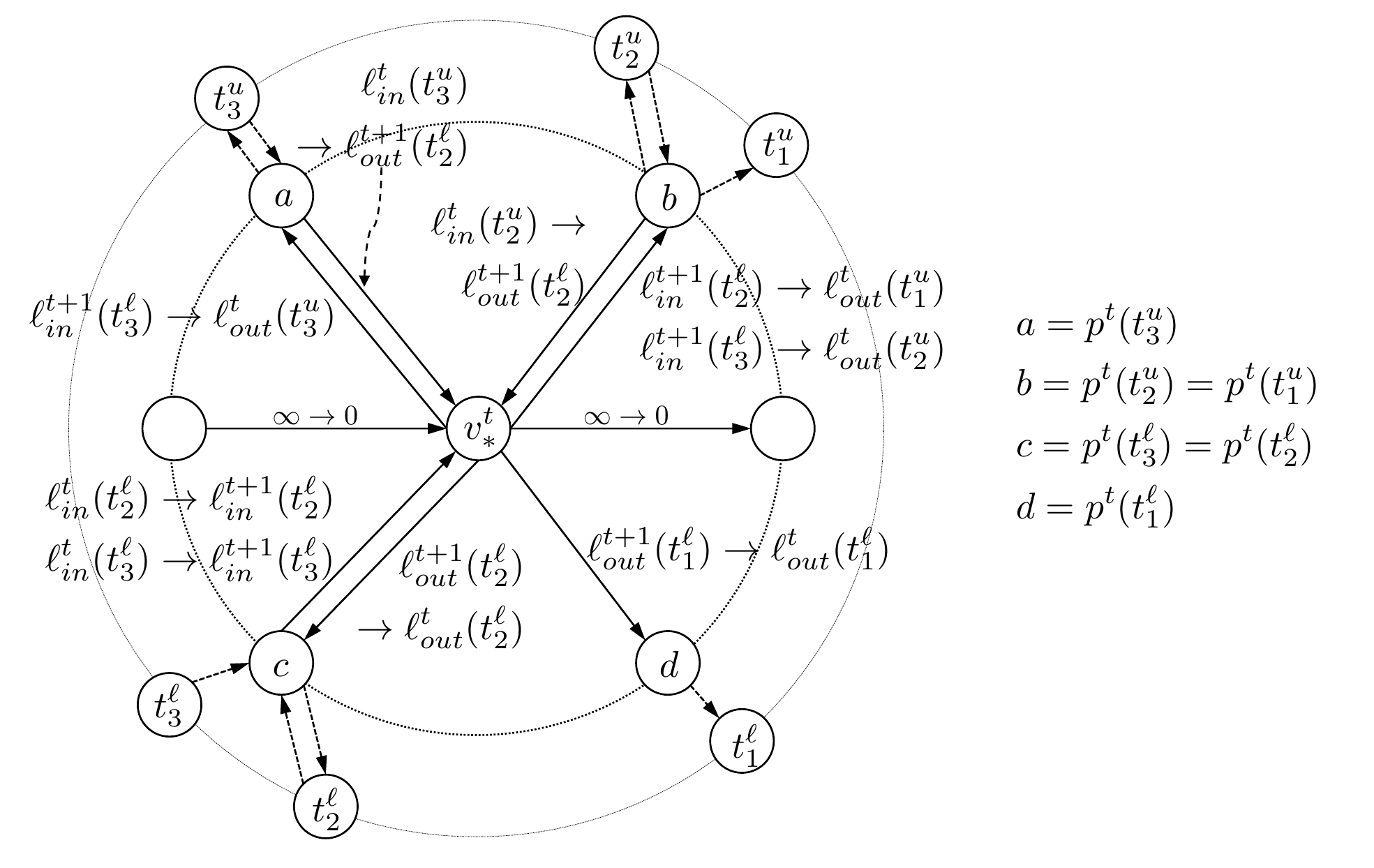}
\caption{How to assign labels to edges.}
\label{fig:label}
\end{center}
\end{figure}

\begin{figure}
\begin{center}
\includegraphics[width= 11 cm]{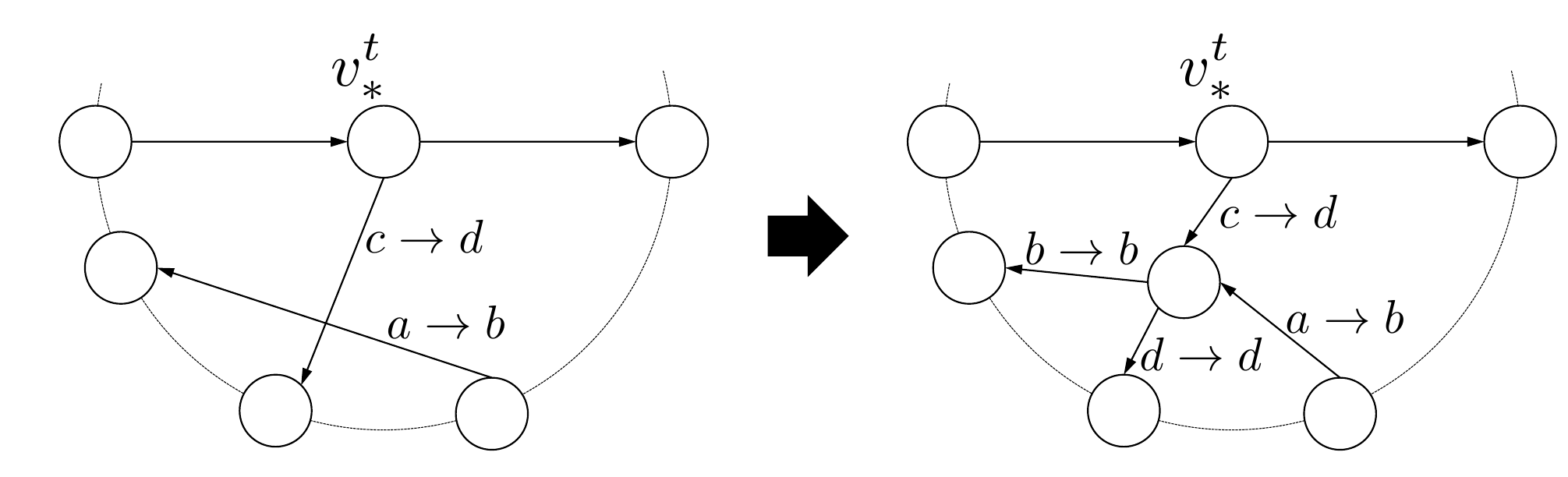}
\caption{How to assign labels to edges made by {\sf MakePlanar}.}
\label{fig:mp_label}
\end{center}
\end{figure}

\begin{figure}
\begin{center}
\includegraphics[width= 10 cm]{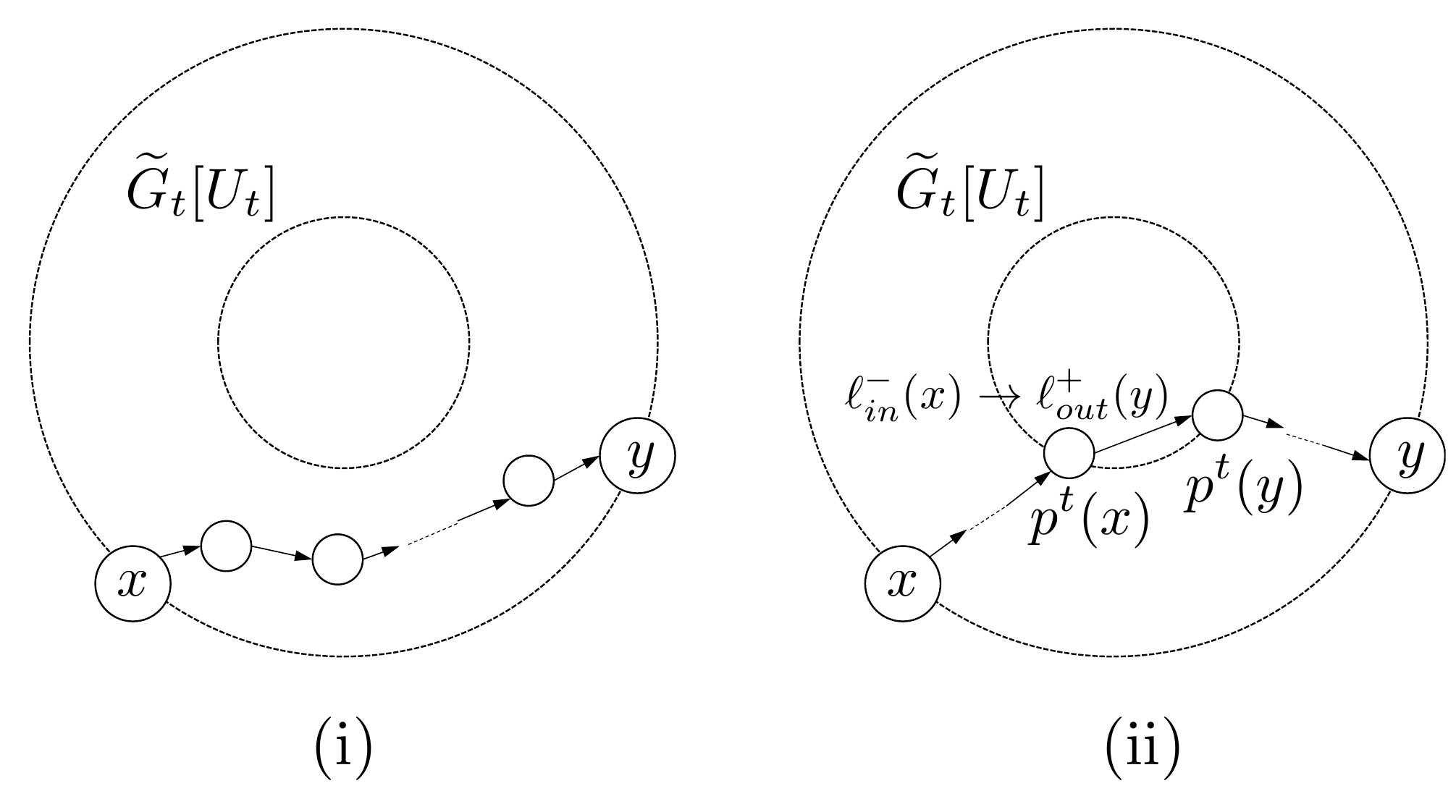}
\caption{Two cases of token tours in the proof of Lemma~\ref{lem:main_lem1}.}
\label{fig:tour_i_ii}
\end{center}
\end{figure}

\begin{lemma}\label{lem:main_lem1}
	For any $t$ in Algorithm~\ref{alg:trans_gadget},
	if there exists an edge from $x$ toward $y$ in $\gadG_0$,
	then there exists a token tour from $x$ to $y$ in $\gadG_t$ whose length is at most $2t+1$.
\end{lemma}

\begin{proof}
We prove the lemma by showing that if $(x, y) \in \gadE_0$ then one of the following two statements holds in $\gadG_t$ for any $t$:
\begin{enumerate}[label=(\roman*)]
\item
there exists a token tour of length at most $2t+1$ from $x$ to $y$ which uses no chords appearing in $\gadG_t[U_t]$ (see Figure~\ref{fig:tour_i_ii}(i)).
\item
there exists a token tour $t_{x, y} = (x, \infty) \Rightarrow \cdots \Rightarrow (p^t(x), \lin^t(x); \lin^t(x)^- \rightarrow \lout^t(y)^+) \Rightarrow (p^t(y), \lout^t(y)^+; \lout^t(y)^-\rightarrow \ell) \Rightarrow \cdots \Rightarrow (y, \infty)$ where $\lin^t(x)^- \le \lin^t(x)$, $\lout^t(y)^+\ge \lout^t(y)$ and $\lout^t(y)^-\le \lout^t(y)$.
In addition, this tour uses no chords appearing in $\gadG_t[U_t]$ except $(p^t(x), p^t(y))$, and its length is at most $2t+1$ (see Figure~\ref{fig:tour_i_ii}(ii)).
\end{enumerate}
We prove by induction on $t$.
We have a tour $(x, \infty; 0\rightarrow \infty) \Rightarrow (y, \infty)$ in $\gadG_0$.
Thus $\gadG_0$ satisfies the statement (i) if $x$ and $y$ are consecutive on the cycle, and otherwise satisfies the statement (ii).

Assume that the statement (i) holds in $\gadG_t$.
The tour from $x$ to $y$ appears also in $\gadG_{t+1}$ and satisfies the statement (i) in $\gadG_{t+1}$. 
Now, we suppose the statement (ii) holds in $\gadG_t$.
We first consider the case that the chord $(p^t(x), p^t(y))$ does not cross $e_*^t$.
When $(p^t(x), p^t(y))$ is in the lower area of $e_*^t$,
the tour $t_{x,y}$ satisfies statement (i) in $\gadG_{t+1}$.
When $(p^t(x), p^t(y))$ is in the upper area of $e_*^t$ or equal to $e_*^t$,
the tour $t_{x,y}$ satisfies statement (ii) in $\gadG_{t+1}$.
Next, we assume that the chord $(p^t(x), p^t(y))$ crosses $e_*^t$.
There are two cases:
\begin{enumerate}[label=(\Roman*)]
	\item
	$p^t(x) \in S^\ell$ and $p^t(y) \in S^u$:
	We have $x\in T^\ell$ and $y\in T^u$.
        There exists a label $\lin^t(x) \rightarrow \lin^{t+1}(x) \in \gadL^{t+1}(p^t(x), v_*^t)$ (cf. line 9 of Algorithm~\ref{alg:change_label}).
        There also exists a label $\ell_{min} \rightarrow \lout^t(y) \in \gadL^{t+1}(v_*^t, p^t(y))$ where $\ell_{min} = \min_{t^\ell \in T^\ell}\{\lin^{t+1}(t^\ell)\ |\ (t^\ell, y) \in \cirblcE\}$ (cf. line 14 of Algorithm~\ref{alg:change_label}).
        Since $x \in T^\ell$ and $(x, y) \in \cirblcE$,
        we have $\ell_{min} \le \lin^{t+1}(x)$.
        Thus, in $\gadG_{t+1}$,
	there exists a token tour $(x, \infty) \Rightarrow \cdots \Rightarrow (p^t(x), \lin^t(x);\lin^t(x)\rightarrow \lin^{t+1}(x)) \Rightarrow (v_*^t, \lin^{t+1}(x);\ell_{min}\rightarrow \lout^t(y)) \Rightarrow (p^t(y), \lout^t(y)) \Rightarrow \cdots \Rightarrow (y, \infty)$.
	If the edge $(p^t(x), v_*^t)$ has a crossing point in the lower area of $e_*^t$,
	we have to modify the part $(p^t(x), \lin^t(x);\lin^t(x)\rightarrow \lin^{t+1}(x)) \Rightarrow (v_*^t, \lin^{t+1}(x))$ to $(p^t(x), \lin^t(x);\lin^t(x)\rightarrow \lin^{t+1}(x)) \Rightarrow (u, \lin^{t+1}(x); \lin^{t+1}(x)\rightarrow \lin^{t+1}(x)) \Rightarrow (v_*^t, \lin^{t+1}(x))$ where $u$ is a vertex created by {\sf MakePlanar}.
	
	\item
	$p^t(x) \in S^u$ and $p^t(y) \in S^\ell$:
	We have $x\in T^u$ and $y\in T^\ell$.
        There exists a label $\lin^t(x) \rightarrow \ell_{max} \in \gadL^{t+1}(p^t(x), v_*^t)$ where $\ell_{max} = \max_{t^\ell \in T^\ell}\{\lout^{t+1}(t^\ell)\ |\ (x, t^\ell) \in \cirblcE\}$ (cf. line 13 of Algorithm~\ref{alg:change_label}).
        Since $y \in T^\ell$ and $(x, y) \in \cirblcE$,
        we have $\ell_{max} \ge \lout^{t+1}(y)$.
        There also exists a label $\lout^{t+1}(y) \rightarrow \lout^t(y) \in \gadL^{t+1}(v_*^t, p^t(y))$ (cf. line 10 of Algorithm~\ref{alg:change_label}).
        Thus, in $\gadG_{t+1}$,
	there exists a token tour $(x, \infty) \Rightarrow \cdots \Rightarrow (p^t(x), \lin^t(x);\lin^t(x)\rightarrow \ell_{max}) \Rightarrow (v_*^t, \ell_{max};\lout^{t+1}(y)\rightarrow \lout^t(y)) \Rightarrow (p^t(y), \lout^t(y)) \Rightarrow \cdots \Rightarrow (y, \infty)$ in $\gadG_{t+1}$.
	If the edge $(v_*^t, p^t(y))$ has a crossing point in the lower area of $e_*^t$,
	we have to modify the part $(v_*^t, \ell_{max};\lout^{t+1}(y))\rightarrow \lout^t(y)) \Rightarrow (p^t(y), \lout^t(y))$ to $(v_*^t, \ell_{max};\lout^{t+1}(y)\rightarrow \lout^t(y)) \Rightarrow (u, \lout^t(y); \lout^t(y)\rightarrow \lout^{t}(y)) \Rightarrow (p^t(y), \lout^{t}(y))$ where $u$ is a vertex created by {\sf MakePlanar}.
\end{enumerate}
In both cases,
the length of the new tour is longer than that of $t_{x,y}$ by at most 2,
thus it is at most $2(t+1)+1$.
We have $p^{t+1}(x) = v_*^t$ in case (I) and $p^{t+1}(y) = v_*^t$ in case (II).
Thus the new tour has only one chord $(p^{t+1}(x), p^{t+1}(y))$ appearing in $\gadG_{t+1}[U_{t+1}]$,
and the chord has a label $\lin^t(x)^- \rightarrow \lout^t(y)^+$.
Therefore the new tour satisfies statement (ii).
\end{proof}

The following lemma shows the other direction:
if there exists a token tour from $x$ to $y$ in the gadget graph,
then there exists a path from $x$ to $y$ in the circle graph.
From Lemma~\ref{lem:cross_reach},
it is enough to prove the following Lemma.

\begin{lemma}\label{lem:main_lem2}
For any $t$ and $x, y \in \cirblcV$,
if there exists a token tour from $x$ to $y$ in $\gadG_t$,
then there exists a traversable edge sequence $(e_1, \ldots, e_k)$ in $\cirblcG$
such that $t(e_1) = x$ and $h(e_k) = y$.
\end{lemma}

Before proving this,
we prepare several lemmas,
i.e., Lemma~\ref{lem:proper_reach} to Lemma~\ref{lem:adjacent}.

We refer to temporal in- and out-levels for $t_i^\ell \in T^\ell$ calculated at line 2 and 3 in Algorithm~\ref{alg:calc_level} as $t\lin^{t+1}(t_i^\ell)$ and $t\lout^{t+1}(t_i^\ell)$.
Namely,
\begin{eqnarray*}
	t\lin^{t+1}(t_i^\ell) &=& \max \{j\ |\ (t_i^\ell, t_j^u) \in \cirblcE,\ t_j^u\in T^u\} + i / n,\\
	t\lout^{t+1}(t_i^\ell) &=& \min \{j\ |\ (t_j^u, t_i^\ell) \in \cirblcE,\ t_j^u\in T^u\} + i / n.
\end{eqnarray*}

\begin{lemma}\label{lem:proper_reach}
In Algorithm~\ref{alg:calc_level} of step $t$,
if $i > j$
then $t\lin^{t+1}(t_i^\ell) > t\lin^{t+1}(t_j^\ell)$,
$\lin^{t+1}(t_i^\ell) > \lin^{t+1}(t_j^\ell)$,
$t\lout^{t+1}(t_i^\ell) > t\lout^{t+1}(t_j^\ell)$ and 
$\lout^{t+1}(t_i^\ell) > \lout^{t+1}(t_j^\ell)$.
Moreover,
if $i > j$ and $t\ell_{io_1}^{t+1}(t_i^\ell) > t\ell_{io_2}^{t+1}(t_j^\ell)$
then
$\ell_{io_1}^{t+1}(t_i^\ell) > \ell_{io_2}^{t+1}(t_j^\ell)$ $(io_1, io_2 \in \{in, out\})$.
\end{lemma}

\begin{proof}
Let $i' = \max\{k\ |\ (t_i^\ell, t_k^u) \in \cirblcE,\ t_k^u\in T^u\}$
and $j' = \max\{k\ |\ (t_j^\ell, t_k^u) \in \cirblcE,\ t_k^u\in T^u\}$.
Namely,
$t\lin^t(t_i^\ell) = i' + i/n$ and
$t\lin^t(t_j^\ell) = j' + j/n$.
Assume $i' < j'$.
Now we have $i > j$ and $i' < j'$.
Thus the edges $(t_i^\ell, t_{i'}^u)$ and $(t_j^\ell, t_{j'}^u)$ are crossing (see Figure~\ref{fig:proper_reach}(a)).
From Lemma~\ref{lem:cross_reach},
the edge $(t_i^\ell, t_{j'}^u)$ is in $\cirblcE$.
This is contrary to the fact that $i'$ is the maximum index.
Thus $t\lin^t(t_i^\ell) = i' + i/n > j' + j/n = t\lin^t(t_j^\ell)$ holds.

Let $i' = \min\{k\ |\ (t_k^u, t_i^\ell) \in \cirblcE,\ t_k^u\in T^u\}$
and $j' = \min\{k\ |\ (t_k^u, t_j^\ell) \in \cirblcE,\ t_k^u\in T^u\}$.
Namely,
$t\lout^t(t_i^\ell) = i' + i/n$ and
$t\lout^t(t_j^\ell) = j' + j/n$.
Assume $i' < j'$.
Now we have $i > j$ and $i' < j'$.
Thus the edges $(t_{i'}^u, t_i^\ell)$ and $(t_{j'}^u, t_j^\ell)$ are crossing (see Figure~\ref{fig:proper_reach}(b)).
From Lemma~\ref{lem:cross_reach},
the edge $(t_{i'}^u, t_j^\ell)$ is in $\cirblcE$.
This is contrary to the fact that $j'$ is the minimum index.
Thus $t\lout^t(t_i^\ell) = i' + i/n > j' + j/n= t\lout^t(t_j^\ell)$ holds.

The in- and out-levels $\lin^t(t_i^\ell)$, $\lin^t(t_j^\ell)$, $\lout^t(t_i^\ell)$ and $\lout^t(t_j^\ell)$
might be larger than $t\lin^t(t_i^\ell)$, $t\lin^t(t_j^\ell)$, $t\lout^t(t_i^\ell)$ and $t\lout^t(t_j^\ell)$
since some positive integer $\Delta$ might be added (see at line 8 and 9 in Algorithm~\ref{alg:calc_level}).
However,
when $\Delta$ is added to $\lin^t(t_j^\ell)$ (resp., $\lout^t(t_j^\ell)$),
$\Delta$ is also added to $\lin^t(t_i^\ell)$ (resp., $\lout^t(t_i^\ell)$)
since $i$ is not less than $j$.
Thus
$\lin^t(t_i^\ell) > \lin^t(t_j^\ell)$ and
$\lout^t(t_i^\ell) > \lout^t(t_j^\ell)$.
By the same argument,
the second statement holds.
\end{proof}

\begin{figure}
\begin{center}
\includegraphics[width= 7 cm]{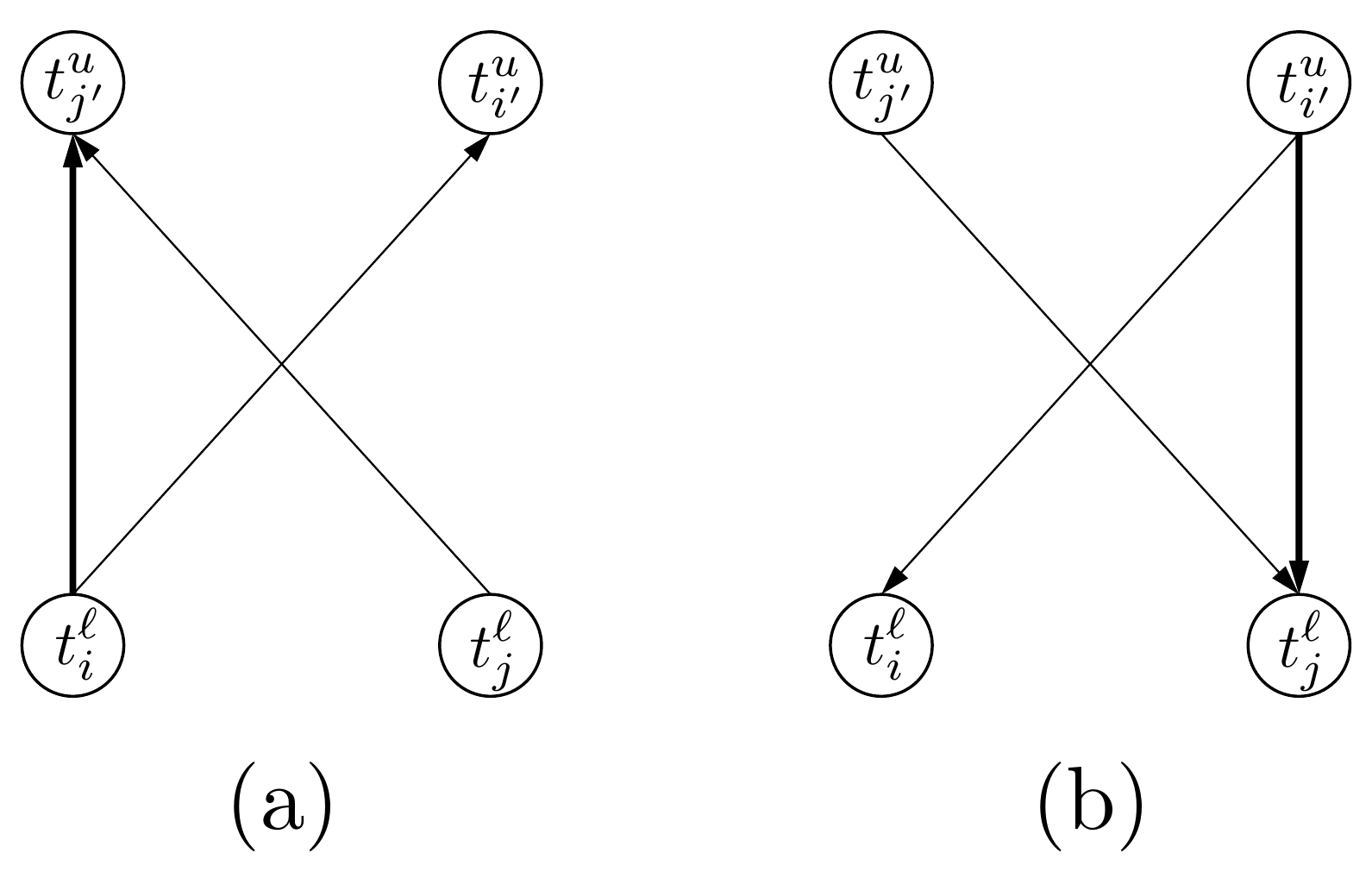}
\caption{Examples of wrong positions of vertices.}
\label{fig:proper_reach}
\end{center}
\end{figure}

\begin{lemma}\label{lem:magnitude}
In Algorithm~\ref{alg:change_label} in step $t$,
for $t_p^\ell$, $t_q^\ell\in T^\ell$,
if $t\lin^{t+1}(t_p^\ell) \ge t\lout^{t+1}(t_q^\ell)$,
then $\lin^{t+1}(t_p^\ell) \ge \lout^{t+1}(t_q^\ell)$.
\end{lemma}

\begin{proof}
When $p > q$,
this lemma holds from Lemma~\ref{lem:proper_reach}.
Consider the case $p \le q$.
Let $t\lin^{t+1}(t_p^\ell) = i+p/n$ and $t\lout^{t+1}(t_q^\ell) = j+q/n$.
Since $t\lin^{t+1}(t_p^\ell) \ge t\lout^{t+1}(t_q^\ell)$ and $p \le q$, we have $i \ge j$.
Assume $\lin^{t+1}(t_p^\ell) < \lout^{t+1}(t_q^\ell)$.
In order that $\lin^{t+1}(t_p^\ell) < \lout^{t+1}(t_q^\ell)$ holds,
some positive integer $\Delta$ should be added to $\lout^{t+1}(t_{q}^\ell)$
at line 9,
and not added to $\lin^{t+1}(t_p^\ell)$
at line 8 of Algorithm~\ref{alg:calc_level}.
Thus,
there should be a vertex $t_{r}^\ell$
such that $q \ge r \ge p$
and $\Delta = \max\{t\lout^{t+1}(t_k^\ell)-k/n\ |\ 1\le k < r\} - (t\lin^{t+1}(t_{r}^\ell)-r/n) > 0$.
Since $\Delta$ is positive,
there exists a vertex $t_{s}^\ell$
such that $r > s$
and $t\lin^{t+1}(t_{r}^\ell) < t\lout^{t+1}(t_{s}^\ell)$.
Since $r \ge p$ and $q \ge s$,
we have $t\lin^{t+1}(t_{r}^\ell) \ge t\lin^{t+1}(t_{p}^\ell)$
and $t\lout^{t+1}(t_q^\ell) \ge t\lout^{t+1}(t_{s}^\ell)$
from Lemma~\ref{lem:proper_reach}.
Now $t\lin^{t+1}(t_{p}^\ell) \ge t\lout^{t+1}(t_q^\ell)$ holds,
thus $t\lin^{t+1}(t_{r}^\ell) \ge t\lout^{t+1}(t_{s}^\ell)$ and
$\Delta$ becomes non-positive.
This is a contradiction.
Thus $\lin^{t+1}(t_p^\ell) \ge \lout^{t+1}(t_q^\ell)$ holds.
\end{proof}

For every label in $\gadG_t$
and $k \le t$,
there are three types:
(i) $\lin^k(x) \rightarrow \lin^{k+1}(x)$ (cf. line 9),
(ii) $\lout^{k+1}(x) \rightarrow \lout^k(x)$ (cf. line 10) and
(iii) $\lin^k(x) \rightarrow \lout^k(y)$ (cf. line 13, 14) for $x, y \in \cirblcV$.
We define a \newwd{source vertex} and a \newwd{sink vertex} for any types of labels.
\begin{enumerate}[label=(\roman*)]
	\item source vertex is $x$. 
	When $i = \max\{j\ |\ (x, t_j^u)\in \cirblcE, t_j^u\in T^u\}$,
	sink vertex is $t_i^u$.
	\item sink vertex is $x$.
	When $i = \min\{j\ |\ (t_j^u, x)\in \cirblcE, t_j^u\in T^u\}$,
	source vertex is $t_i^u$.
	\item source vertex is $x$ and sink vertex is $y$.
\end{enumerate}
For a label $L$,
we refer to an edge in $\cirblcG$ from $L$'s source vertex to $L$'s sink vertex as a \newwd{source edge} of $L$.
It is obvious that
any source edge exists in $\cirblcG$.

\begin{lemma}\label{lem:location}
We consider any token tour of length 2 going through $v_*^t$$:$
$(x, a'; a\rightarrow b) \Rightarrow (v_*^t, b; c\rightarrow d) \Rightarrow (y, d)$ in $\gadG_{t'}$ where $t < t'$.
\begin{enumerate}
\item $(x, v_*^t)$ is in upper area, and $(v_*^t, y)$ is in upper area of $e_*^t$:
Let $(t_i^u, t_p^\ell)$ be 
$a\rightarrow b$'s source edge,
and we let $(t_q^\ell, t_j^u)$ be
$c\rightarrow d$'s source edge.
We have $p \ge q$.

\item $(x, v_*^t)$ is in upper area, and $(v_*^t, y)$ is in lower area of $e_*^t$:
Let $(t_i^u, t_p^\ell)$ be 
$a\rightarrow b$'s source edge,
and we let $(t_j^u, t_q^\ell)$ be
$c\rightarrow d$'s source edge.
We have $p \ge q$.

\item $(x, v_*^t)$ is in lower area, and $(v_*^t, y)$ is in upper area of $e_*^t$:
Let $(t_p^\ell, t_i^u)$ be 
$a\rightarrow b$'s source edge,
and we let $(t_q^\ell, t_j^u)$ be
$c\rightarrow d$'s source edge.
We have $i \ge j$ and $p \ge q$.

\item $(x, v_*^t)$ is in lower area, and $(v_*, y)$ is in lower area of $e_*^t$:
Let $(t_p^\ell, t_i^u)$ be 
$a\rightarrow b$'s source edge,
and we let $(t_j^u, t_q^\ell)$ be
$c\rightarrow d$'s source edge.
We have
$({\rm i})$ $i \ge j$ and $p \ge q$,
$({\rm ii})$ $i \ge j$ and $p < q$ or
$({\rm iii})$ $i < j$ and $p \ge q$.

\end{enumerate}
The indices $i$, $j$, $p$ and $q$ are
based on the sequences $T^u$ and $T^\ell$ made in Algorithm~\ref{alg:change_label} in step $t$.
\end{lemma}

\begin{proof}
\ 
\begin{enumerate}
\item
We have $b = \lout^{t+1}(t_p^\ell)$ and $c = \lin^{t+1}(t_q^\ell)$.
From the rule of token tours,
$\lout^{t+1}(t_p^\ell) \ge \lin^{t+1}(t_q^\ell)$ holds.
If $p < q$,
we have $\lout^{t+1}(t_p^\ell) < \lin^{t+1}(t_q^\ell)$
from Lemma~\ref{lem:in_out}.
Thus we have $p \ge q$.

\item 
We have $b = \lout^{t+1}(t_p^\ell)$ and $c = \lout^{t+1}(t_q^\ell)$.
From the rule of token tours,
$\lout^{t+1}(t_p^\ell) \ge \lout^{t+1}(t_q^\ell)$ holds.
If $p < q$,
we have $\lout^{t+1}(t_p^\ell) < \lout^{t+1}(t_q^\ell)$ from Lemma~\ref{lem:proper_reach}.
Thus we have $p \ge q$.

\item
We have $b = \lin^{t+1}(t_p^\ell)$ and $c = \lin^{t+1}(t_q^\ell)$.
From the rule of token tours,
$\lin^{t+1}(t_p^\ell) \ge \lin^{t+1}(t_q^\ell)$ holds.
If $p < q$,
we have $\lin^{t+1}(t_p^\ell) < \lin^{t+1}(t_q^\ell)$ from Lemma~\ref{lem:proper_reach}.
Thus we have $p \ge q$.
From the definition of source edge,
$t_i^u$ has the maximum index among vertices that $t_p^\ell$ can reach. 
Assume $i < j$.
Now we have $p \ge q$ and $i < j$.
Thus the edges $(t_p^\ell, t_i^u)$ and $(t_q^\ell, t_j^u)$ are semi-crossing.
From Lemma~\ref{lem:cross_reach},
the edge $(t_p^\ell, t_j^u)$ is in $\cirblcE$.
This is contrary to the fact that $i$ is the maximum index.
Thus $i \ge j$ holds.

\item 
We have $b = \lin^{t+1}(t_p^\ell)$ and $c = \lout^{t+1}(t_q^\ell)$.
From the rule of token tours,
$\lin^{t+1}(t_p^\ell) \ge \lout^{t+1}(t_q^\ell)$ holds.
We will show that $i < j$ and $p < q$ do not hold simultaneously.
Assume $i < j$ and $p < q$.
We have $i+p/n = t\lin^{t+1}(t_p^\ell) < t\lout^{t+1}(t_q^\ell) = j + q/n$.
From Lemma~\ref{lem:proper_reach},
we have $\lin^{t+1}(t_p^\ell) < \lout^{t+1}(t_q^\ell)$ since $p < q$,
and we cannot follow the tour.
Thus,
there are three possible relationships:
(i) $i \ge j$ and $p \ge q$,
(ii) $i \ge j$ and $p < q$,
(iii) $i < j$ and $p \ge q$.
\end{enumerate}
\end{proof}

\begin{figure}
\begin{center}
\includegraphics[width= 10 cm]{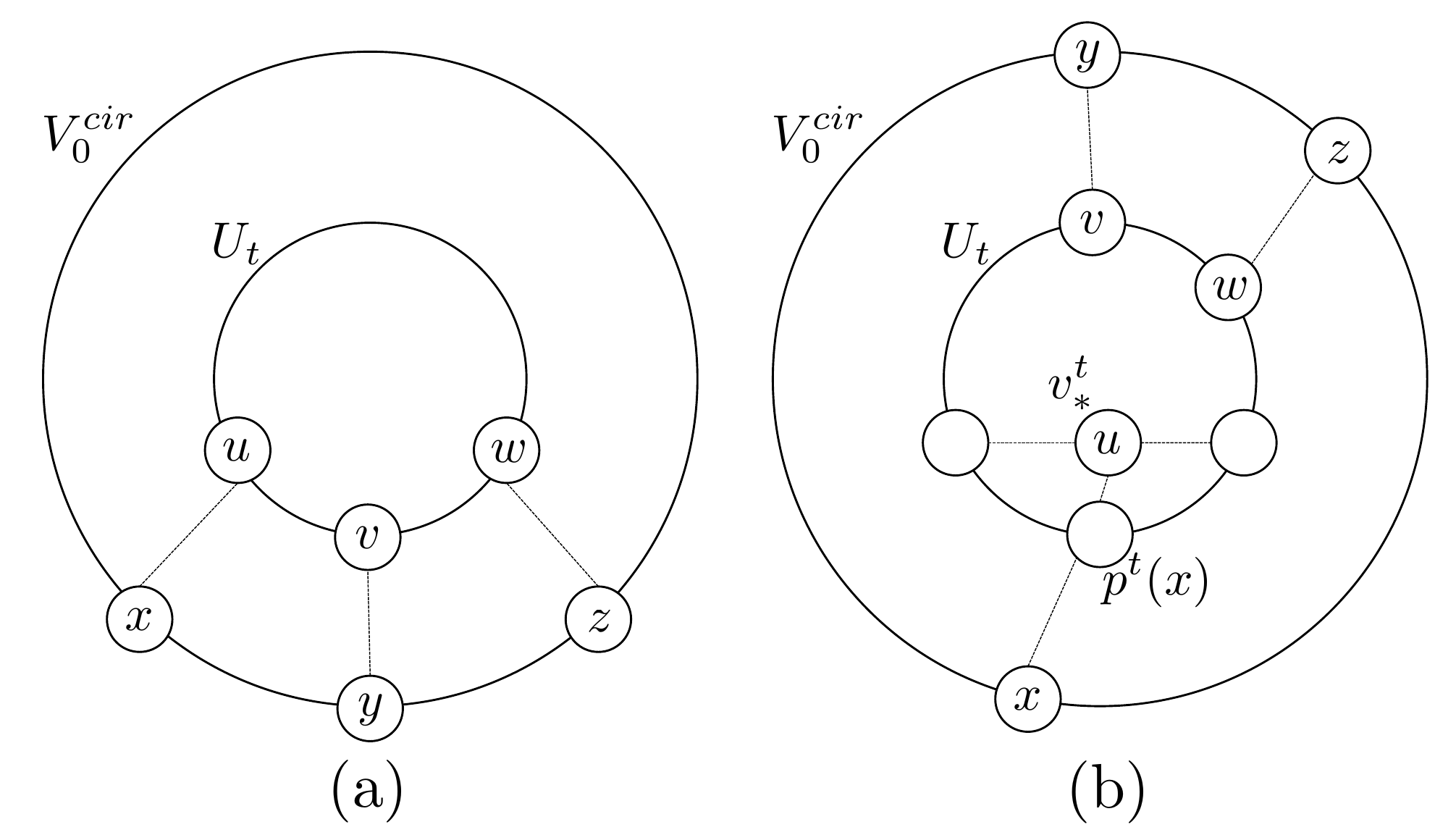}
\caption{Relations of vertices and their parents.}
\label{fig:parent_child}
\end{center}
\end{figure}

\begin{lemma}\label{lem:proper_order}
  For any three vertices $u$, $v$, $w \in U_t$,
  if $(u, v, w)$ is in clockwise (resp., anti-clockwise) order in $\gadG_t[U_t]$,
  then
  $(x, y, z)$ is also in clockwise (resp., anti-clockwise) order in $\cirblcG$
  for any $x$, $y$, $z \in \cirblcV$
  such that $p^t(x) = u$, $p^t(y) = v$ and $p^t(z) = w$ $($see $Figure~\ref{fig:parent_child} ({\rm a}))$.
\end{lemma}

\begin{proof}
We prove by induction on $t$.
Since the parent of every vertex is itself in step 0,
the Lemma is true in step 0.
Let $u$, $v$ and $w$ be vertices such that
$(u, v, w)$ is in clockwise (resp., anti-clockwise) order in $\gadG_{t+1}[U_{t+1}]$.
Fix any three vertices $x$, $y$ and $z$
such that $p^{t+1}(x) = u$, $p^{t+1}(y) = v$ and $p^{t+1}(z) = w$.
When none of $u$, $v$ or $w$ is $v_*^t$,
$(u, v, w)$ is in clockwise (resp., anti-clockwise) order 
also in $\gadG_{t}[U_{t}]$.
Thus $(x, y, z)$ is in clockwise (resp., anti-clockwise) order
from the induction hypothesis.
Let $u = v_*^t$.
Since $p^t(x)$ is in the lower area of $e_*^t$,
$(p^t(x), u, v)$ is in clockwise (resp., anti-clockwise) order in $\gadG_t[U_t]$ (see Figure~\ref{fig:parent_child}(b)).
From the induction hypothesis,
$(x, y, z)$ is in clockwise (resp., anti-clockwise) order in $\cirblcG$.
In cases $v = v_*^t$ or $w = v_*^t$,
the Lemma is proved in the same way.
\end{proof}

\begin{lemma}\label{lem:adjacent}
  For any $k \le t$ such that $v_*^k \in U_t$,
  let $x$ and $y$ be vertices such that $p^k(x) = t(e_*^k)$ and $p^k(y) = h(e_*^k)$ respectively.
  $p^t(x)$, $v_*^k$ and $p^t(y)$ are consecutive in $\gadG_t[U_t]$.
\end{lemma}

\begin{proof}
Fix any $k$.
We prove by induction on $t$.
When $t = k$,
it is obvious that $p^t(x)$, $v_*^k$ and $p^t(y)$ are consecutive.
Assume the Lemma is true for a fixed $t$.
If $v_*^k$ is not an endpoint of $e_*^{t}$,
we have $p^t(x) = p^{t+1}(x)$ and $p^t(y) = p^{t+1}(y)$.
Thus $p^{t+1}(x)$, $v_*^k$ and $p^{t+1}(y)$ are consecutive in $\gadG_{t+1}[U_{t+1}]$ from the induction hypothesis.
When $v_*^k$ is an endpoint of $e_*^{t}$ and
$p^t(x)$ (resp., $p^t(y)$) is on $C_{\gadG_t[U_t]}(t(e_*^t), h(e_*^t))$,
we have $p^{t+1}(x) = v_*^t$ (resp., $p^{t+1}(y) = v_*^t$).
Since $v_*^k$ and $v_*^t$ are adjacent,
$p^{t+1}(x)$, $v_*^k$ and $p^{t+1}(y)$ are consecutive in $\gadG_{t+1}[U_{t+1}]$.
\end{proof}

Now
we prove Lemma~\ref{lem:main_lem2}.
\setcounter{lemma}{5}
\begin{lemma}[restated]
For any $t$ and $x, y \in \cirblcV$,
if there exists a token tour from $x$ to $y$ in $\gadG_t$,
then there exists a traversable edge sequence $(e_1, \ldots, e_k)$ in $\cirblcG$
such that $t(e_1) = x$ and $h(e_k) = y$.
\end{lemma}
\setcounter{lemma}{11}

\begin{proof}
Let $t_{x,y}$ be a token tour $(v_1, \ell_1; f_1\rightarrow \ell_2) \Rightarrow (v_2, \ell_2; f_2 \rightarrow \ell_3) \Rightarrow \cdots \Rightarrow (v_m, \ell_m)$
such that $v_1 = x$ and $v_m = y$.
We modify $t_{x,y}$.
If the edges $(v_i, v_{i+1}),\ldots,(v_{i+d-1}, v_{i+d})$ are made by {\sf MakePlanar} and they have the same original edge,
namely $\gadK_t(v_j, v_{j+1}) = (v_{j+1}, v_{j+2})$ $(i\le j < i+d-1)$,
we change the partial tour $(v_i, \ell_i; f_i\rightarrow \ell_{i+1}) \Rightarrow (v_{i+1}, \ell_{i+1}; f_{i+1}\rightarrow \ell_{i+2}) \Rightarrow \cdots \Rightarrow (v_{i+d}, \ell_{i+d})$
to $(v_i, \ell_i; f_i\rightarrow \ell_{i+1}) \Rightarrow (v_{i+d}, \ell_{i+1})$.
Note that $\ell_{i+1}$ is equal to $\ell_{i+d}$.
Next,
we remove redundant moves.
Consider a partial tour of length 2 $(v_i, \ell_i; f_i\rightarrow \ell_{i+1}) \Rightarrow (v_{i+1}, \ell_{i+1}; f_{i+1}\rightarrow \ell_{i+2}) \Rightarrow (v_{i+2}, \ell_{i+2}; f_{i+2}\rightarrow \ell_{i+3})$.
When $v_i = v_{i+2}$ and $f_i \ge \ell_{i+2}$,
we regard this move as a redundant move.
We have $\ell_i \ge f_i$ and $\ell_{i+2} \ge f_{i+2}$.
If the move is redundant, $\ell_i \ge f_{i+2}$ holds.
We change the partial tour
$(v_i, \ell_i; f_i\rightarrow \ell_{i+1}) \Rightarrow (v_{i+1}, \ell_{i+1}; f_{i+1}\rightarrow \ell_{i+2}) \Rightarrow (v_{i+2}, \ell_{i+2}; f_{i+2}\rightarrow \ell_{i+3}) \Rightarrow (v_{i+3}, \ell_{i+3})$
 to $(v_i, \ell_i; f_{i+2}\rightarrow \ell_{i+3}) \Rightarrow (v_{i+3}, \ell_{i+3})$.
Again,
we let the changed token tour be $t_{x,y} = (v_1, \ell_1; f_1\rightarrow \ell_2) \Rightarrow (v_2, \ell_2; f_2 \rightarrow \ell_3) \Rightarrow \cdots \Rightarrow (v_{m}, \ell_{m})$.

We construct a traversable edge sequence.
We put source edges of the labels appearing in $t_{x,y}$.
Let this edge sequence be $(e_1,\ldots, e_{m-1})$
where $e_i$ is a source edge of the label $f_i\rightarrow \ell_{i+1}$.
For each $i$ $(2\le i < m)$,
$v_i$ corresponds to $v_*^k$ for some $1 \le k < t$
since we removed vertices made by {\sf MakePlanar} from the tour.
For every $i$ $(2\le i < m)$,
we take an edge $e'_i \in \cirblcE$
such that $p^k(t(e'_i)) = t(e_*^k)$ and $p^k(h(e'_i)) = h(e_*^k)$
where $k$ is derived from $v_i = v_*^k$.
There exist several ways to choose $e'_i$.
We show that
if we select $e'_i$'s appropriately,
the edge sequence $(e_1, e'_2, e_2, e'_3, e_3,\ldots, e_{m-2}, e'_{m-1}, e_{m-1})$ becomes traversable.

\begin{figure}
\begin{center}
\includegraphics[width= 12 cm]{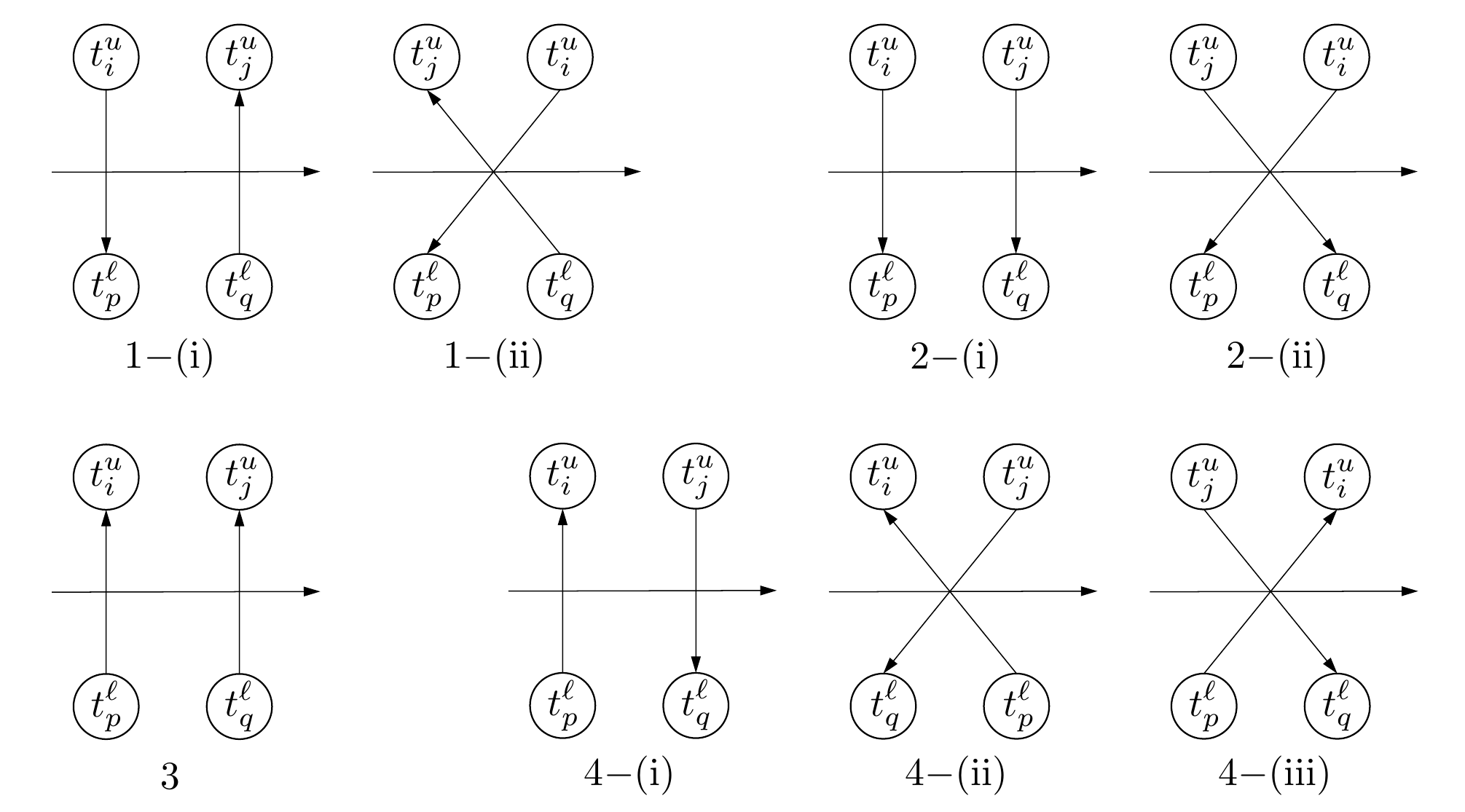}
\caption{Relations of two source edges.}
\label{fig:basic_relation}
\end{center}
\end{figure}

It is obvious that $e_1$ crosses $e'_2$.
We have to show that all edges except for $e_1$ and $e'_2$ separate two edges appearing before themselves.
By induction,
suppose we fixed $e'_j$ $(i+2 \le j < m)$.
We show which pair of edges $e_{i+1}$ separates.
Consider the partial tour of length 2 $(v_i, \ell_i; f_i\rightarrow \ell_{i+1}) \Rightarrow (v_{i+1}, \ell_{i+1}; f_{i+1}\rightarrow \ell_{i+2}) \Rightarrow (v_{i+2}, \ell_{i+2})$,
and let $v_{i+1} = v_*^k$.
We suppose that this partial tour corresponds to case 1 of Lemma~\ref{lem:location}, namely the edge $(v_i, v_{i+1})$ is in upper area and the edge $(v_{i+1}, v_{i+2})$ is in upper area of $e_*^k$.
We let $(t_i^u, t_p^\ell)$ be 
$f_i\rightarrow \ell_{i+1}$'s source edge,
and $(t_q^\ell, t_j^u)$ be
$f_{i+1}\rightarrow \ell_{i+2}$'s source edge.
From Lemma~\ref{lem:location},
we have $p \ge q$.
Thus $e_{i+1} = (t_q^\ell, t_j^u)$ separates $e'_{i+1}$ and $e_i=(t_i^u, t_p^\ell)$
for any choice of $e'_{i+1}$
(see Figure~\ref{fig:basic_relation} 1-(i), 1-(ii).
The horizontal edge corresponds to $e'_{i+1}$. 
In both cases (i)$i\ge j$ and (ii)$i < j$,
$t_p^\ell$ and $h(e'_{i+1})$ are on the opposite side of $(t_q^\ell, t_j^u)$).
When the partial tour corresponds to case 2, 3, 4-(i) or 4-(ii),
$e_{i+1}$ separates $e'_{i+1}$ and $e_i$ for any choice of $e'_{i+1}$ (see Figure~\ref{fig:basic_relation}).

Suppose the partial tour corresponds to case 4-(iii).
Assume $v_i = v_{i+2}$.
Let $v_i$ and $v_{i+1}$ be $v_*^k$ and $v_*^t$ respectively.
Let $(t_p^\ell, t_i^u)$ be $f_i\rightarrow \ell_{i+1}$'s source edge
and $(t_j^u, t_q^\ell)$ be $f_{i+1}\rightarrow \ell_{i+2}$'s source edge.
The indices $i$, $j$, $p$ and $q$ are
based on the sequences $T^u$ and $T^\ell$ made in Algorithm~\ref{alg:change_label} in step $t$.
Recall $i < j$ and $p \ge q$.
We will show that $f_i = \lin^{k+1}(t_p^\ell) \ge \lout^{k+1}(t_q^\ell) = \ell_{i+2}$.
Suppose $e_*^k$ and $e_*^t$ has the same direction,
namely $t_p^\ell$ and $t_q^\ell$ had indices $p'$ and $q'$ respectively in step $k$ such that $p' \ge q'$.
In this case,
we have $\lin^{k+1}(t_p^\ell) \ge \lout^{k+1}(t_q^\ell)$ from Lemma~\ref{lem:in_out}.
Assume $e_*^k$ and $e_*^t$ has the opposite direction,
namely $t_i^u$, $t_j^u$, $t_p^\ell$ and $t_q^\ell$ had indices $i'$, $j'$, $p'$ and $q'$ respectively in step $k$ such that $i' > j'$ and $p' \le q'$.
Since $i' > j'$, $t\lin^{k+1}(t_p^\ell)-p/n \ge i'$ and $t\lout^{k+1}(t_q^\ell)-q/n \le j'$,
we have $t\lin^{k+1}(t_p^\ell) \ge t\lout^{k+1}(t_q^\ell)$.
From Lemma~\ref{lem:magnitude},
$\lin^{k+1}(t_p^\ell) \ge \lout^{k+1}(t_q^\ell)$.
Thus,
when $v_i = v_{i+2}$,
this is a redundant move and does not appear in $t_{x,y}$.

\begin{figure}
\begin{center}
\includegraphics[width= 14 cm]{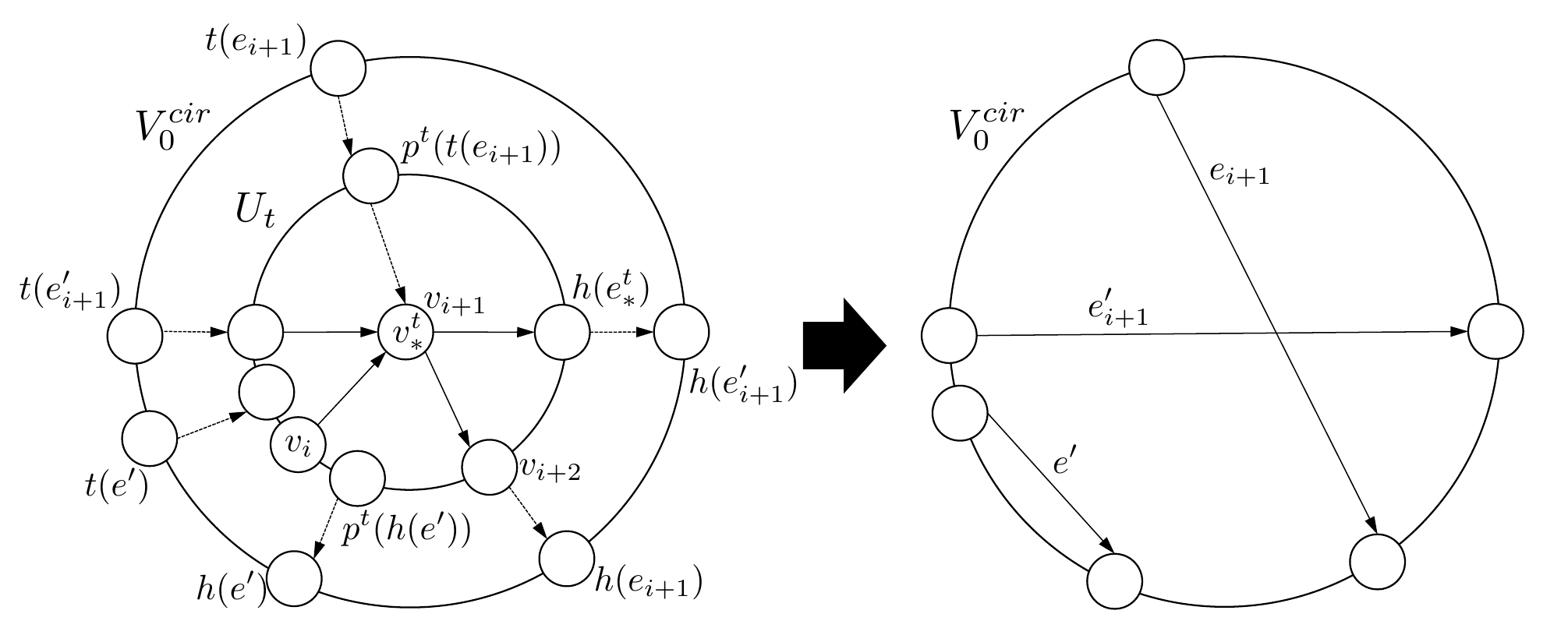}
\caption{case 4-(iii): $p^t(h(e'))$ is on $C_{\gadG_t[U_t]}[t(e_*^t), v_{i+2})$.}
\label{fig:case4-3-1}
\end{center}
\end{figure}

\begin{figure}[h]
\begin{center}
\includegraphics[width= 14 cm]{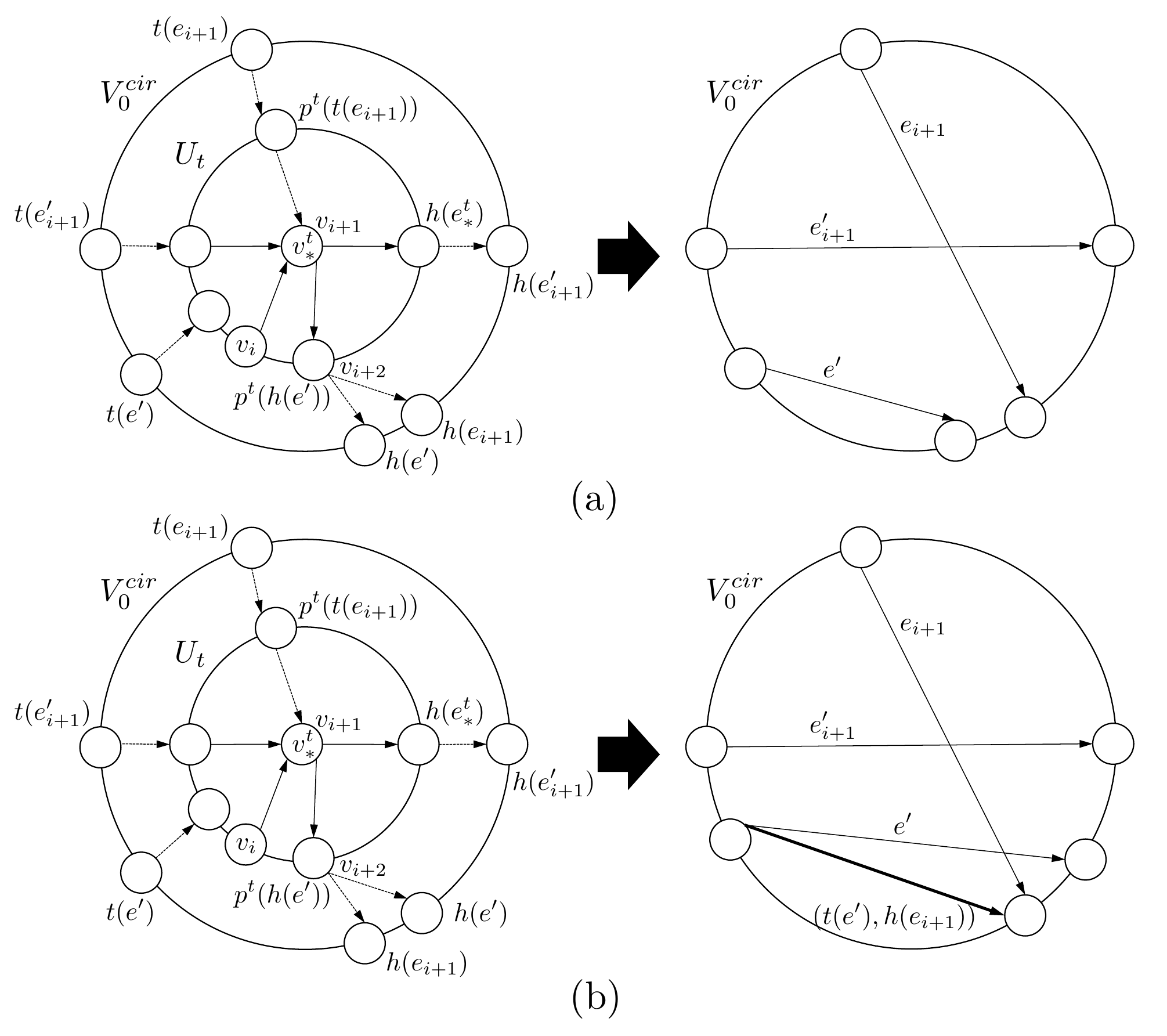}
\caption{case 4-(iii): $p^t(h(e'))$ is equal to $v_{i+2}$.}
\label{fig:case4-3-2}
\end{center}
\end{figure}

\begin{figure}[t]
\begin{center}
\includegraphics[width= 14 cm]{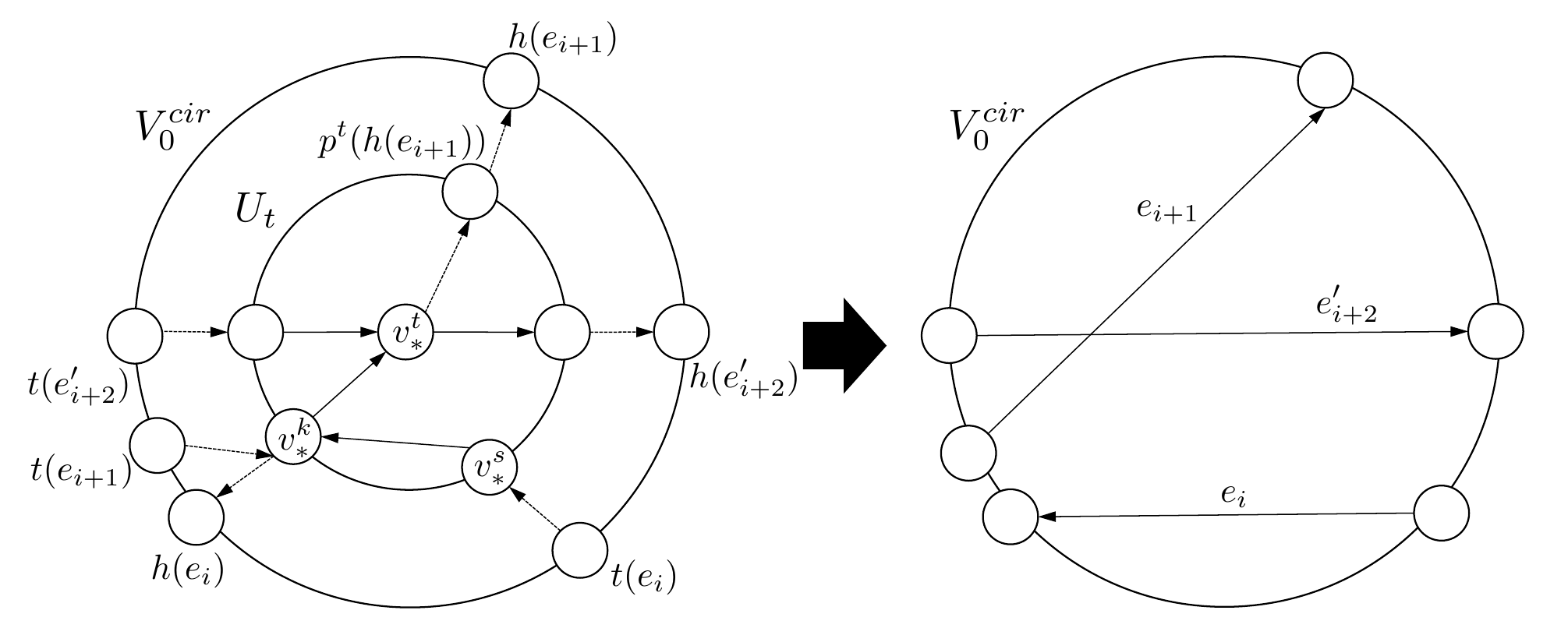}
\caption{$t > k$: $v_*^s$ is on $U_t$.}
\label{fig:t_k_1}
\end{center}
\end{figure}

Again,
let $v_i$ and $v_{i+1}$ be $v_*^k$ and $v_*^t$ respectively.
Let $e'$ be an edge such that $p^k(t(e')) = t(e_*^k)$ and $p^k(h(e')) = h(e_*^k)$.
Now we consider the case $v_i\neq v_{i+2}$,
thus
there are two cases for a location of $p^t(h(e'))$ from Lemma~\ref{lem:adjacent}.
\begin{enumerate}
\item
$p^t(h(e'))$ is on $C_{\gadG_t[U_t]}[t(e_*^t), v_{i+2})$:
Consider the four vertices $p^t(h(e'))$, $v_{i+2}$, $h(e_*^t)$ and $p^t(t(e_{i+1}))$.
A possible order on $U_t$ of the four vertices is
$(p^t(h(e')), v_{i+2}, h(e_*^t), p^t(t(e_{i+1})))$.
From Lemma~\ref{lem:proper_order},
$e_{i+1}$ separates $e'$ and $e'_{i+1}$ for any choice of $e'$ (see Figure~\ref{fig:case4-3-1}).

\item
$p^t(h(e'))$ is equal to $v_{i+2}$:
If $e_{i+1}$ separates $e'$ and $e'_{i+1}$,
we set $e'$ as $e'_{i}$ (see Figure~\ref{fig:case4-3-2} (a)).
However,
$e_{i+1}$ might not separate $e'$ and $e'_{i+1}$.
In this case,
$e'$ crosses $e_{i+1}$.
From Lemma~\ref{lem:cross_reach},
there exists an edge $(t(e'), h(e_{i+1}))$ in $\cirblcE$ (see Figure~\ref{fig:case4-3-2} (b)).
We choose $(t(e'), h(e_{i+1}))$ as $e'_{i}$ and
$e_{i+1}$ separates $e'_{i}$ and $e'_{i+1}$ in $\cirblcG$.

\end{enumerate}

Suppose we fixed $e'_j$ $(i+2 \le j < m)$.
We show how to select $e'_{i+1}$ and
which pair of edges $e'_{i+2}$ separates.
Consider any partial token tour of length 2
$(v_i, \ell_i; f_i\rightarrow \ell_{i+1}) \Rightarrow (v_{i+1}, \ell_{i+1}; f_{i+1}\rightarrow \ell_{i+2}) \Rightarrow (v_{i+2}, \ell_{i+2})$.
Assume $v_i = v_*^s$, $v_{i+1} = v_*^k$ and $v_{i+2} = v_*^t$.
Note that $e_{i}$ and $e_{i+1}$ are source edges of $f_i\rightarrow \ell_{i+1}$ and $f_{i+1}\rightarrow \ell_{i+2}$ respectively,
and
$e'_{i+2}$ is an edge such that $p^t(t(e'_{i+2})) = t(e_*^t)$ and $p^t(h(e'_{i+2})) = h(e_*^t)$.
We let $e'$ be a source edge of a label of $e_*^k$.
When $t > k$:

\begin{figure}
\begin{center}
\includegraphics[width= 14 cm]{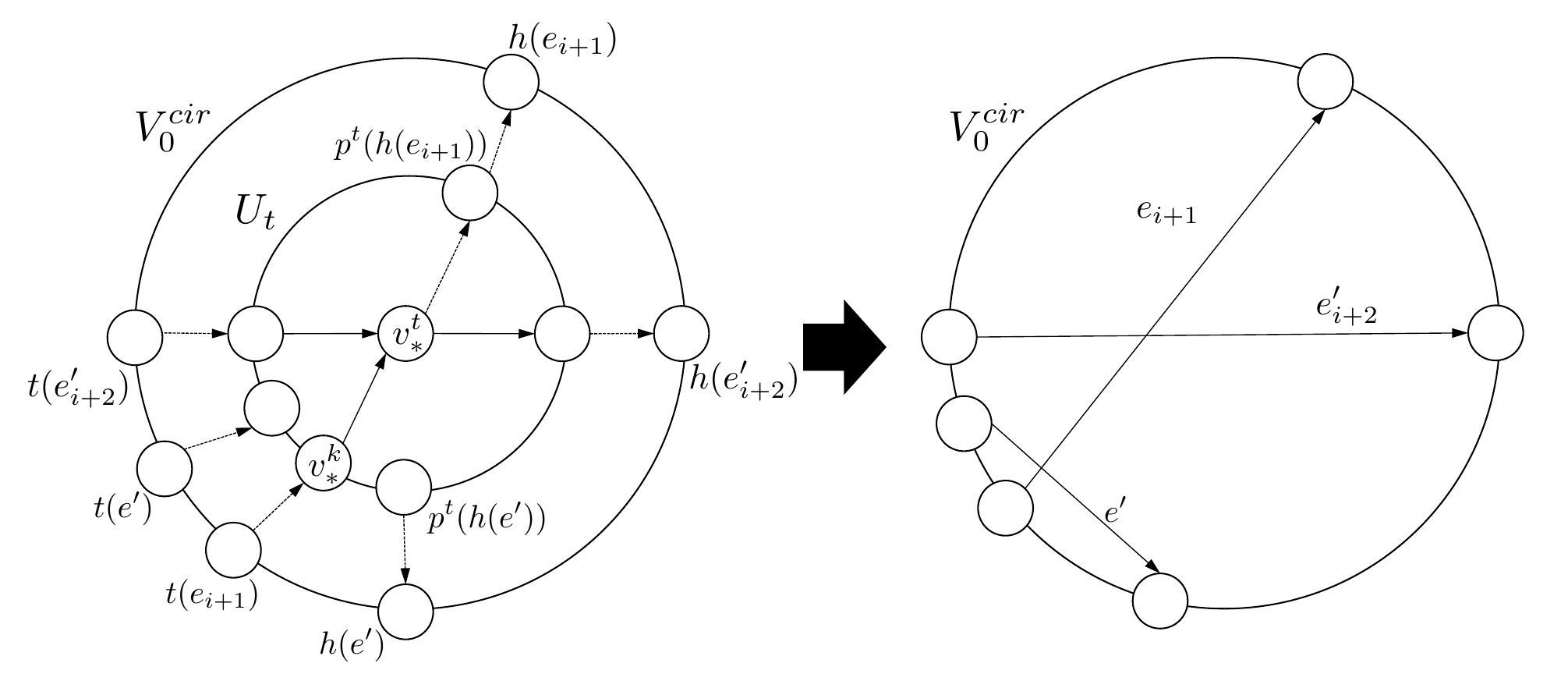}
\caption{$t > k$: $v_*^s$ is not on $U_t$, $p^t(h(e'))$ is neither $h(e_*^t)$ nor $t(e_*^t)$.}
\label{fig:t_k_2_1}
\end{center}
\end{figure}

\begin{figure}
\begin{center}
\includegraphics[width= 14 cm]{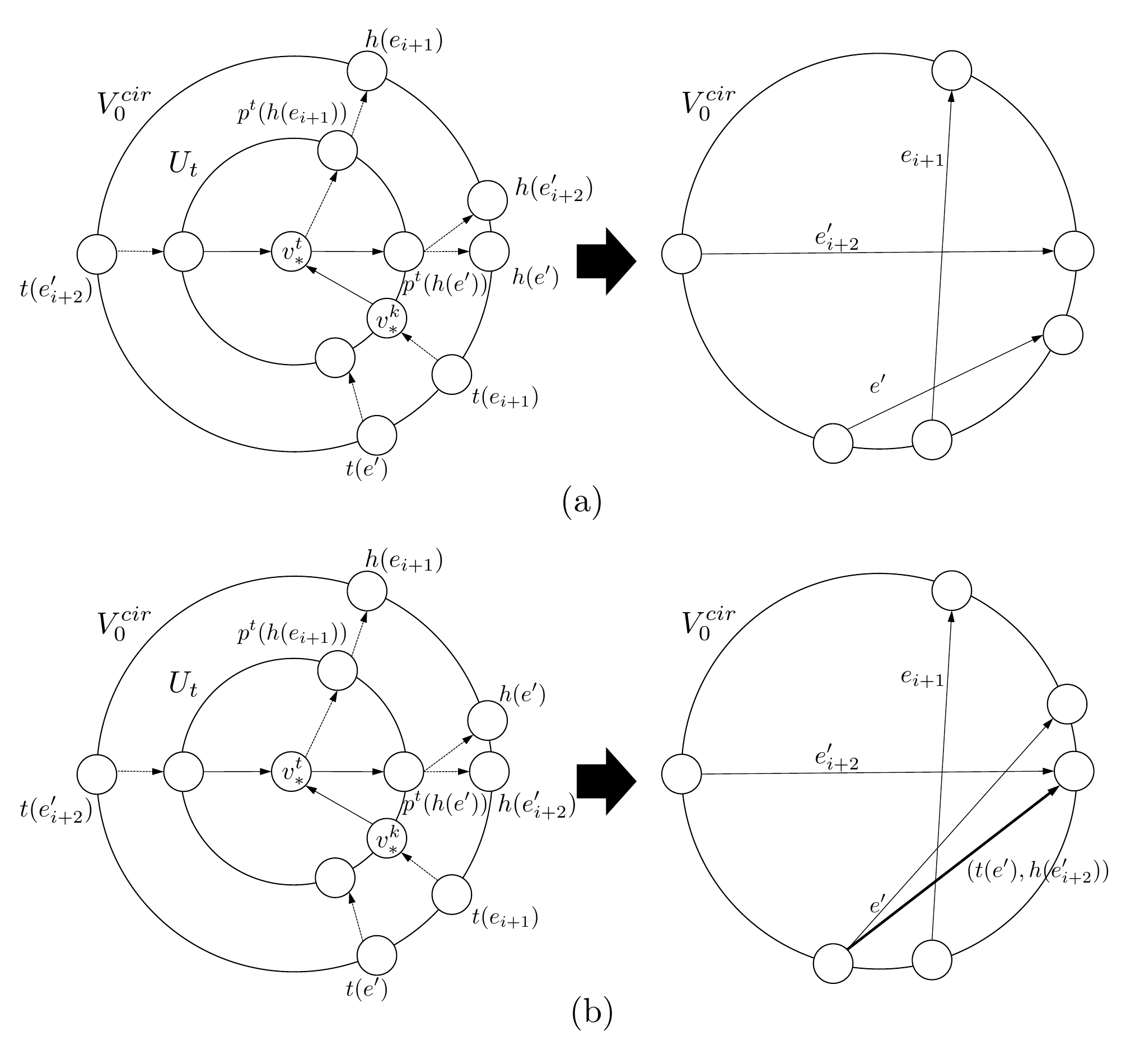}
\caption{$t > k$: $v_*^s$ is not on $U_t$, $p^t(h(e'))$ is equal to $h(e_*^t)$.}
\label{fig:t_k_2_2}
\end{center}
\end{figure}

\begin{figure}
\begin{center}
\includegraphics[width= 14 cm]{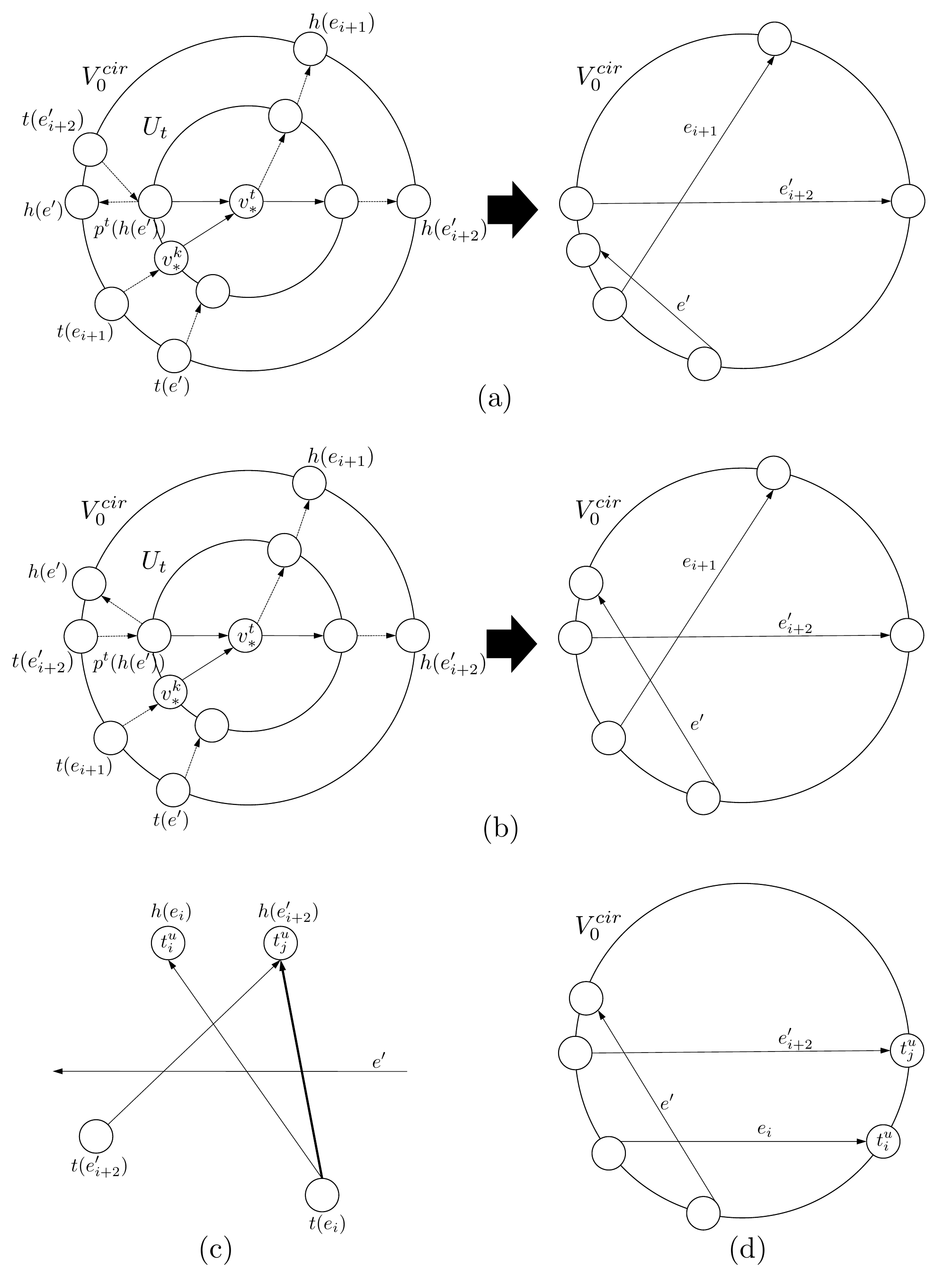}
\caption{$t > k$: $v_*^s$ is not on $U_t$, $p^t(h(e'))$ is equal to $t(e_*^t)$.}
\label{fig:t_k_2_3}
\end{center}
\end{figure}

\begin{enumerate}
\item
$v_*^s$ is on $U_t$:
Consider the four vertices $p^t(t(e'_{i+2}))$, $v_*^k$, $p^t(h(e'_{i+2}))$ and $p^t(h(e_{i+1}))$.
A possible order on $U_t$ of the four vertices is
$(p^t(t(e'_{i+2})), v_*^k, p^t(h(e'_{i+2})), p^t(h(e_{i+1})))$.
From Lemma~\ref{lem:proper_order},
$e'_{i+2}$ separates $e_{i+1}$ and $e_{i}$ (see Figure~\ref{fig:t_k_1}).
We select $e'$ as $e'_{i+1}$.

\item
$v_*^s$ is not on $U_t$:
From Lemma~\ref{lem:adjacent},
there are three cases for a location of $p^t(h(e'))$.
\begin{enumerate}[label=(\roman*)]
\item
$p^t(h(e'))$ is neither $h(e_*^t)$ nor $t(e_*^t)$:
Consider the four vertices $p^t(t(e'_{i+2}))$, $p^t(h(e'))$, $p^t(h(e'_{i+2}))$ and $p^t(h(e_{i+1}))$.
A possible order on $U_t$ of the four vertices is\\
$(p^t(t(e'_{i+2})), p^t(h(e')), p^t(h(e'_{i+2})), p^t(h(e_{i+1})))$.
From Lemma~\ref{lem:proper_order},
$e'_{i+2}$ separates $e_{i+1}$ and $e'$ (see Figure~\ref{fig:t_k_2_1}).
We select $e'$ as $e'_{i+1}$.

\item
$p^t(h(e'))$ is equal to $h(e_*^t)$:
When $e'_{i+2}$ separates $e_{i+1}$ and $e'$ (see Figure~\ref{fig:t_k_2_2}(a)),
we select $e'$ as $e'_{i+1}$.
$e'_{i+2}$ might not separate $e_{i+1}$ and $e'$ (see Figure~\ref{fig:t_k_2_2}(b)).
In this case,
$e'$ crosses $e'_{i+2}$.
From Lemma~\ref{lem:cross_reach},
there exists an edge $(t(e'), h(e'_{i+2}))$ in $\cirblcE$.
We choose $(t(e'), h(e'_{i+2}))$ as $e'_{i+1}$ and
$e'_{i+2}$ separates $e_{i+1}$ and $e'_{i+1}$ in $\cirblcG$.

\item
$p^t(h(e'))$ is equal to $t(e_*^t)$:
When $e'_{i+2}$ separates $e_{i+1}$ and $e'$ (see Figure~\ref{fig:t_k_2_3}(a)),
we select $e'$ as $e'_{i+1}$.
$e'_{i+2}$ might not separate $e_{i+1}$ and $e'$ (see Figure~\ref{fig:t_k_2_3}(b)).
In step $k$,
$h(e_{i})$ and $h(e'_{i+2})$ was in $T^u$.
Let $i$ and $j$ be indices of $h(e_i)$ and $h(e'_{i+2})$ respectively, that is, $h(e_i) = t_i^u$ and $h(e'_{i+2}) = t_j^u$.
If $i < j$,
then $e_{i}$ crosses $e'_{i+2}$,
and $t(e_{i})$ can reach $h(e'_{i+2})$ (see Figure~\ref{fig:t_k_2_3}(c)).
Since $t_i^u$ has the maximum index among vertices $t(e_{i})$ can reach,
this is a contradiction.
Therefore we have $i \ge j$.
Thus $e'_{i+2}$ separates $e_{i}$ and $e'$ (see Figure~\ref{fig:t_k_2_3}(d)).
We select $e'$ as $e'_{i+1}$.

\end{enumerate}
\end{enumerate}

\begin{figure}[t]
\begin{center}
\includegraphics[width= 14 cm]{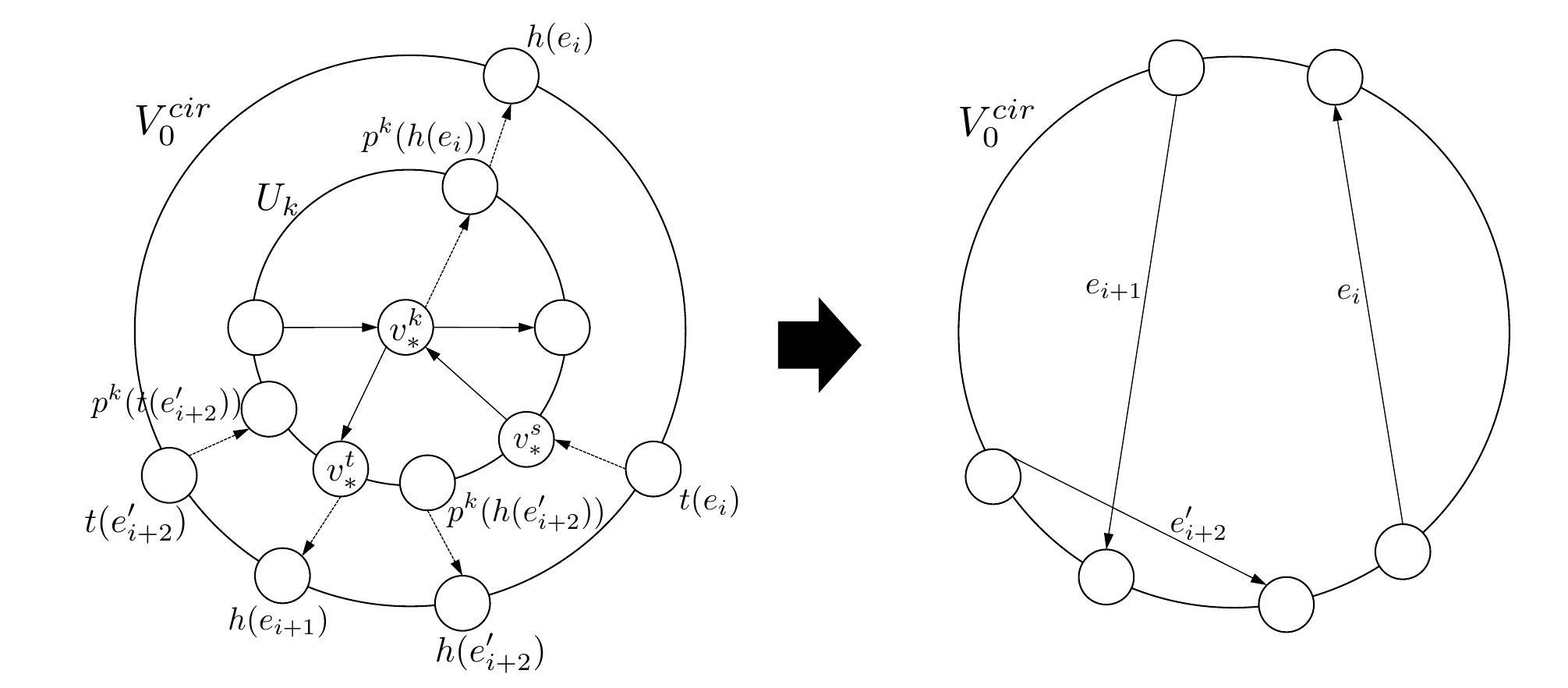}
\caption{$k > t$: $(v_*^s, v_*^k)$ and $(v_*^k, v_*^t)$ are in the same side of $e_*^k$.}
\label{fig:k_t_1}
\end{center}
\end{figure}

\begin{figure}
\begin{center}
\includegraphics[width= 14 cm]{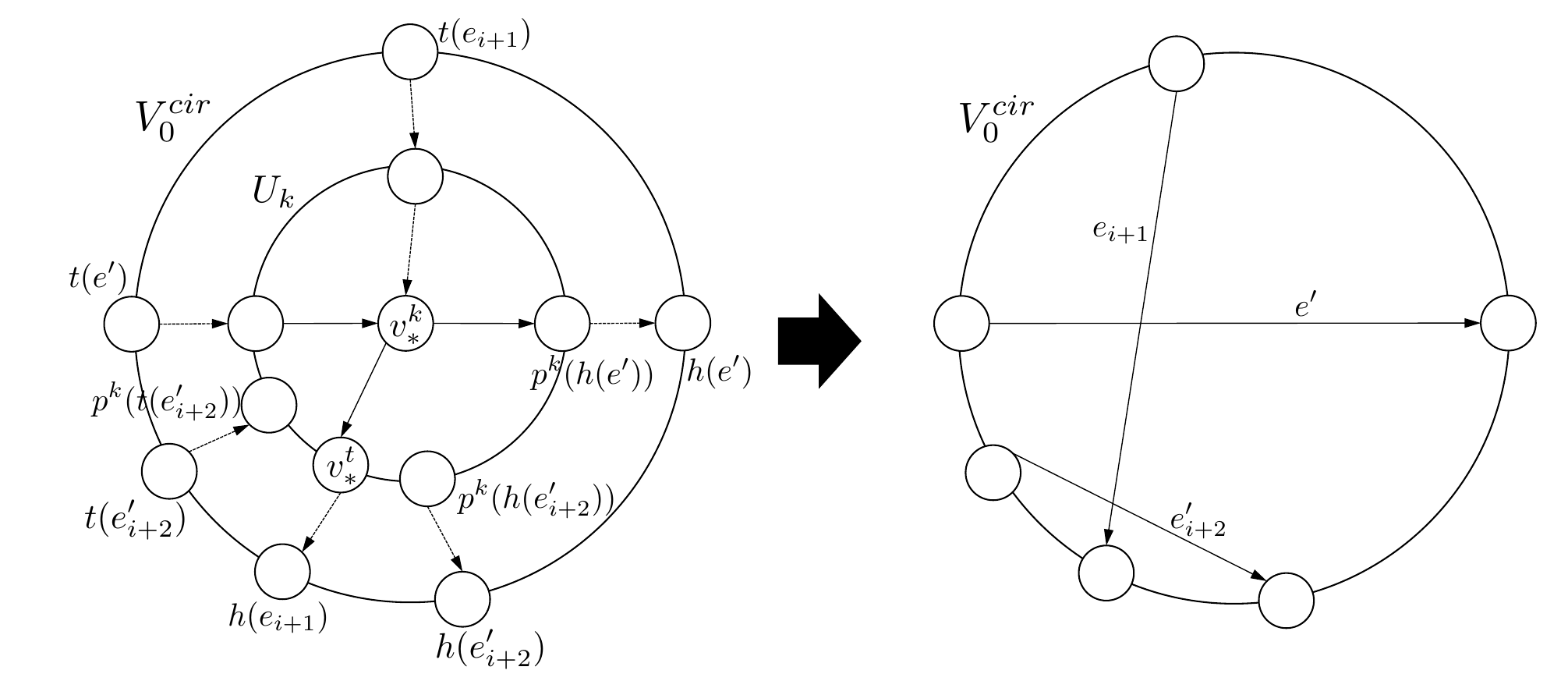}
\caption{$k > t$: $(v_*^s, v_*^k)$ and $(v_*^k, v_*^t)$ are in the opposite side of $e_*^k$, neither $p^k(h(e'_{i+2}))$ nor $p^k(t(e'_{i+2}))$ is $h(e_*^k)$.}
\label{fig:k_t_2_1}
\end{center}
\end{figure}

\begin{figure}[t]
\begin{center}
\includegraphics[width= 14 cm]{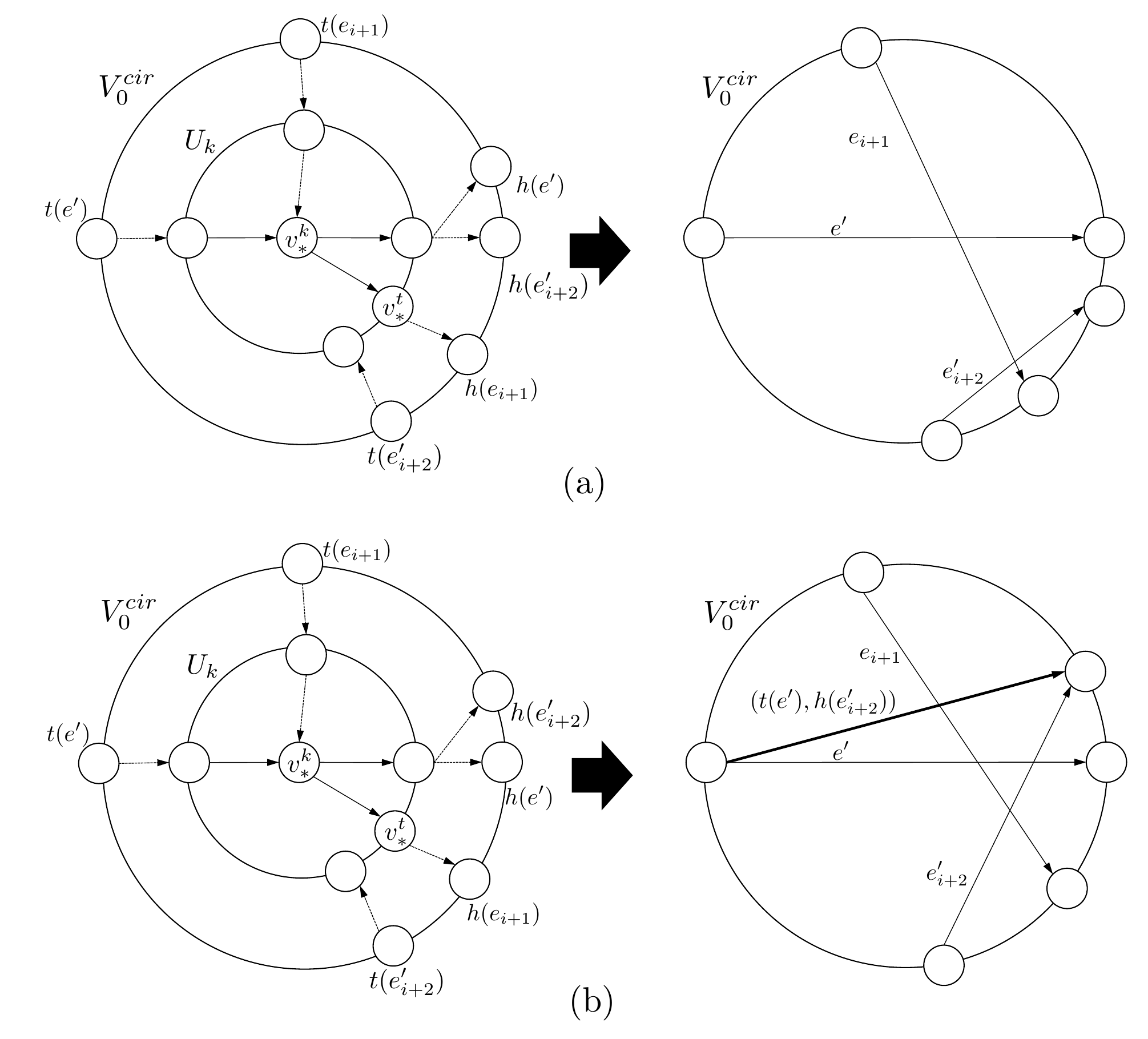}
\caption{$k > t$:  $(v_*^s, v_*^k)$ and $(v_*^k, v_*^t)$ are in the opposite side of $e_*^k$, $p^k(h(e'_{i+2}))$ is equal to $h(e_*^k)$.}
\label{fig:k_t_2_2}
\end{center}
\end{figure}

\begin{figure}
\begin{center}
\includegraphics[width= 14 cm]{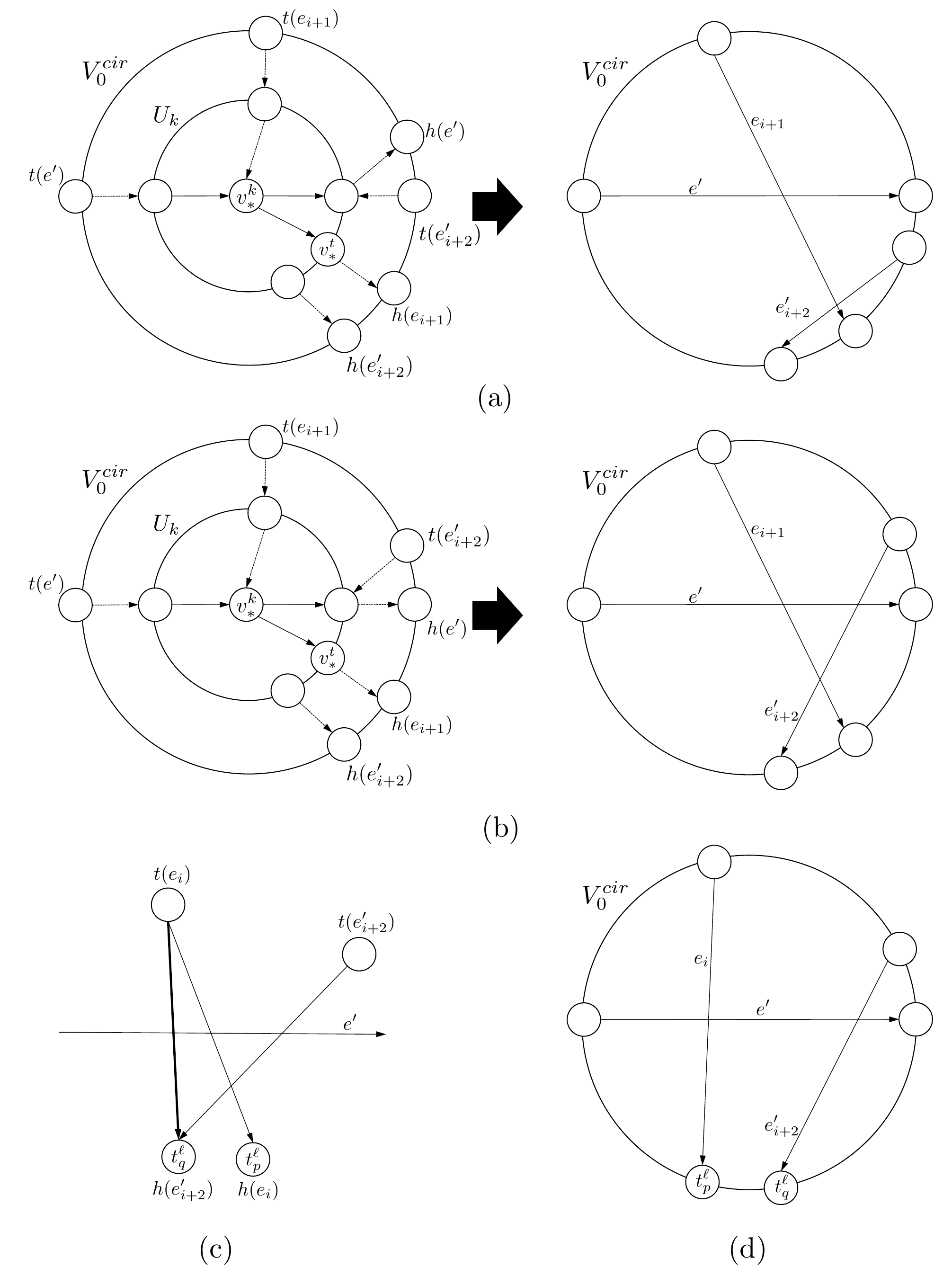}
\caption{$k > t$: $(v_*^s, v_*^k)$ and $(v_*^k, v_*^t)$ are in the opposite side of $e_*^k$, $p^k(t(e'_{i+2}))$ is equal to $h(e_*^k)$.}
\label{fig:k_t_2_3}
\end{center}
\end{figure}

When $k > t$:
\begin{enumerate}
\item
$(v_*^s, v_*^k)$ and $(v_*^k, v_*^t)$ are in the same side of $e_*^k$:
$p^k(t(e'_{i+2}))$, $p^k(h(e_{i+1}))$ and $p^k(h(e'_{i+2}))$ are consecutive on $U_k$ from Lemma~\ref{lem:adjacent}.
The parent of $h(e_{i})$ is on the other side of $e_*^k$.
From Lemma~\ref{lem:proper_order},
$e'_{i+2}$ separates $e_{i+1}$ and $e_{i}$ (see Figure~\ref{fig:k_t_1}).
We select $e'$ as $e'_{i+1}$.

\item
$(v_*^s, v_*^k)$ and $(v_*^k, v_*^t)$ are in the opposite side of $e_*^k$:
From Lemma~\ref{lem:adjacent},
there are three cases for a location of $p^k(h(e'_{i+2}))$ and $p^k(t(e'_{i+2}))$.
\begin{enumerate}[label=(\roman*)]
\item
Neither $p^k(h(e'_{i+2}))$ nor $p^k(t(e'_{i+2}))$ is $h(e_*^k)$:
Consider the four vertices $p^k(t(e'_{i+2}))$, $p^k(h(e'_{i+2}))$, $p^k(h(e_{i+1}))$ and $p^k(e'))$.
A possible order on $U_k$ of the four vertices is\\
$(p^k(t(e'_{i+2})), p^k(h(e_{i+1})), p^k(h(e'_{i+2})), p^k(e'))$ (see Figure~\ref{fig:k_t_2_1}).
From Lemma~\ref{lem:proper_order},
$e'_{i+2}$ separates $e_{i+1}$ and $e'$.
We select $e'$ as $e'_{i+1}$.

\item
$p^k(h(e'_{i+2}))$ is equal to $h(e_*^k)$:
When $e'_{i+2}$ separates $e_{i+1}$ and $e'$ (see Figure~\ref{fig:k_t_2_2}(a)),
we select $e'$ as $e'_{i+1}$.
$e'_{i+2}$ might not separate $e_{i+1}$ and $e'$ (see Figure~\ref{fig:k_t_2_2}(b)).
In this case,
$e'$ crosses $e'_{i+2}$.
From Lemma~\ref{lem:cross_reach},
there exists an edge $(t(e'), h(e'_{i+2}))$ in $\cirblcE$.
We choose $(t(e'), h(e'_{i+2}))$ as $e'_{i+1}$, and
$e'_{i+2}$ separates $e_{i+1}$ and $e'_{i+1}$ in $\cirblcG$.

\item
$p^k(t(e'_{i+2}))$ is equal to $h(e_*^k)$:
When $e'_{i+2}$ separates $e_{i+1}$ and $e'$ (see Figure~\ref{fig:k_t_2_3}(a)),
we select $e'$ as $e'_{i+1}$.
$e'_{i+2}$ might not separate $e_{i+1}$ and $e'$ (see Figure~\ref{fig:k_t_2_3}(b)).
Assume the edge $(v_*^s, v_*^k)$ is in the upper area of $e_*^k$.
In step $k$,
$h(e_{i})$ and $h(e'_{i+2})$ was in $T^\ell$.
Let $p$ and $q$ be indices of $h(e_i)$ and $h(e'_{i+2})$ respectively, that is, $h(e_i) = t_p^\ell$ and $h(e'_{i+2}) = t_q^\ell$.
If $p < q$,
then $e_{i}$ crosses $e'_{i+2}$,
and $t(e_{i})$ can reach $h(e'_{i+2})$ (see Figure~\ref{fig:k_t_2_3}(c)).
Since $t_p^\ell$ has the maximum index among vertices $t(e_{i})$ can reach,
this is a contradiction.
Therefore we have $p \ge q$.
Thus $e'_{i+2}$ separates $e_{i}$ and $e'$ (see Figure~\ref{fig:k_t_2_3}(d)).
We select $e'$ as $e'_{i+1}$.
In the case $(v_*^s, v_*^k)$ is in the lower area of $e_*^k$,
we could show that $e'_{i+2}$ separates $e_{i}$ and $e'$
in the almost same way.

\end{enumerate}
\end{enumerate}

When we consider which pair of edges $e_{i+1}$ separates,
we might choose a specific $e'_i$ (see Figure~\ref{fig:case4-3-2}(b)).
When we consider which pair of edges $e'_{i+1}$ separates,
we also might choose specific $e'_i$ (see Figure~\ref{fig:t_k_2_2}(b) and \ref{fig:k_t_2_2}(b)).
If the cases of  Figure~\ref{fig:case4-3-2}(b) and Figure~\ref{fig:t_k_2_2}(b) occur simultaneously,
$h(e_*^t)$ must be $v_{i+2}$,
but the edge $(v_*^t, t(e_*^t))$ has no available label.
In the case of Figure~\ref{fig:k_t_2_2}(b),
the edge $e_i$ is in the upper area of $e_*^t$.
Thus, these cases never occur simultaneously.
From the above,
the constructed edge sequence is traversable.
\end{proof}

We analyze the space and time complexity of Algorithm~\ref{alg:trans_gadget}.
Note that, 
for saving computation space,
we do not implement the Algorithm straightforwardly in some points.
We begin with the space complexity.
We regard the circle graph $\cirblcG = (\cirblcV, \cirblcE)$ as the input.
For every $v \in \cirblcV$,
we keep three attributes $\lin^t(v)$, $\lout^t(v)$ and $p^t(v)$
in step $t$.
The in- and out-levels are rational numbers that have the form of $i + j/n$.
Thus we keep two integers $i$ and $j$ for each in- and out-level.
We use $\widetilde{O}(n)$ space for preserving them.
In step $t$,
we also keep $U_t$ by using $\widetilde{O}(n)$ space.
We need $\gadG_t[U_t]$,
but we do not keep $\gadE_t$ explicitly.
For $u, v \in U_t$,
whether there exists an edge $(u, v)$ in $\gadG_t[U_t]$ is equivalent to
whether there exists an edge $(x, y)$ in $\cirblcE$ such that $p^t(x) = u$ and $p^t(y) = v$.
Since $\cirblcE$ is included in the input,
we could calculate it with $\widetilde{O}(1)$ space.
We keep no other information throughout the Algorithm.
The number of edges in $\gadG_t[U_t]$ is at most $2|U_t|^2 = O(n^2)$.
Thus,
for line 4 and 5,
we can find a lowest gap-$2^+$ chord by using $\widetilde{O}(1)$ space.
For line 7 and 9,
we use only $\widetilde{O}(1)$ space for updating $\gadV_t$ and $U_t$.
For line 8,
we ignore the edges in the upper area of $e_*^t$ (these edges belong to $\gadG_{t+1}[U_{t+1}]$, thus we have no need to keep them).
For the edges in the lower area of $e_*^t$,
since there exist only gap-$1$ chords in the area,
the number of edges in the area is $O(n)$.
We use $\widetilde{O}(n)$ space for temporarily keeping them.
In {\sf MakePlanar} (line 10),
we look through them,
and find crossing points and resolve them and set $\gadK_{t+1}(\cdot)$ by using $\widetilde{O}(n)$ space.

Now we consider Algorithm~\ref{alg:change_label}.
The number of edges in $\gadG_t[U_t]$ is at most $2|U_t|^2 = O(n^2)$.
Thus,
for line 2,
we can find $S^\ell$ and $S^u$ by using $\widetilde{O}(1)$ space,
and
we use $\widetilde{O}(n)$ space for keeping them.
For line 3 to 6,
since $|T^\ell|, |T^u| = O(n)$,
we also use $\widetilde{O}(n)$ for keeping $T^\ell$ and $T^u$.
In addition,
we use $\widetilde{O}(n)$ space for calculating $\lin^{t+1}(v)$ and $\lout^{t+1}(v)$ for all $v\in T^\ell$.
In Algorithm~\ref{alg:calc_level},
we use $\widetilde{O}(1)$ space for each operation and 
the length of for-loops is $O(n)$.
Thus we use $\widetilde{O}(1)$ space in all.
For line 8 to 11,
we only refer to in- and out-levels that we are keeping.
For line 12 to 15,
we do not keep and ignore the labels belonging to edges in the upper area.
For line 16,
we use $\widetilde{O}(1)$ space.
For line 17 to 24,
we check whether an edge in the lower area was divided by {\sf MakePlanar}
and we use additional $\widetilde{O}(1)$ space.
For line 25 to 27,
we can find all vertices in the lower area of $e_*^t$ by using $\widetilde{O}(n)$ space,
and
we use additional $\widetilde{O}(1)$ space for updating $p^{t+1}(\cdot)$.

We go back to Algorithm~\ref{alg:trans_gadget}.
For line 12,
we output the information of the vertices, edges, labels and values of  the path function in the lower area of $e_*^t$.
Here we have to calculate the labels on the gap-$1$ chords
(other information is preserved now).
Let the gap-$1$ chord be $(v_*^p, v_*^q)$.
If $p < q$,
this edge was made in step $q$ and
the labels on the edge were calculated at line 13 of Algorithm~\ref{alg:change_label}.
Thus,
for any $v \in \cirblcV$ such that $p^t(v) = v_*^p$,
we calculate $\lout = \max_{t^\ell \in \cirblcV, p^t(t^\ell) = v_*^q}\{\lout^{t}(t^\ell)\ |\ (v, t^\ell) \in \cirblcE\}$,
and $\lin^t(v) \rightarrow \lout$ becomes one of the labels on the edge
(if the vertex $v$ is not in $T^u$,
$\lout$ is not defined
and a label for $v$ does not exist).
On the other hand,
if $p > q$,
this edge was made in step $p$ and
the labels on the edge were calculated at line 14 of Algorithm~\ref{alg:change_label}.
Thus,
for any $v \in \cirblcV$ such that $p^t(v) = v_*^q$,
we calculate $\lin = \min_{t^\ell \in \cirblcV, p^t(t^\ell) = v_*^p}\{\lin^{t}(t^\ell)\ |\ (t^\ell, v) \in \cirblcE\}$,
and $\lin \rightarrow \lout^t(v)$ becomes one of the labels on the edge
(if the vertex $v$ is not in $T^u$,
$\lin$ is not defined
and a label for $v$ does not exist).
We use additional $\widetilde{O}(1)$ space for these calculation.
For line 15,
we trace line 10 to 12.
In total,
we use $\widetilde{O}(n)$ space.

Next consider the time complexity.
In Lemma~\ref{lem:trans_gadget_small},
we proved that the while-loop at line 4 stops after at most $n$ steps.
Since the sizes of $U_t$, $S^\ell$, $S^u$, $T^\ell$ and $T^u$ are all $O(n)$,
every operation in the Algorithms takes poly($n$) time.
Thus
this algorithm runs in polynomial time.

\begin{lemma}\label{lem:alg_complexity}
Algorithm~\ref{alg:trans_gadget} runs in polynomial time with using $\widetilde{O}(n)$ space.
\end{lemma}

From Lemma~\ref{lem:main_lem1}, \ref{lem:main_lem2} and \ref{lem:alg_complexity},
we can obtain desired $\gadplG$ with $\tldO(n) = \tldO(N^{1/3})$ space and polynomial time.

\section{Apply {\sf PlanarReach} to a Gadget Graph}
By applying {\sf PlanarReach} to the obtained plane gadget graph $\gadplG$ with $O(N^{2/3})$ vertices,
we can prove Theorem~\ref{thm:main}.
In this section,
we explain how to apply {\sf PlanarReach} to a plane gadget graph,
which has labels in edges.
We have to modify {\sf PlanarReach} slightly.
We now describe the outline of the algorithm {\sf PlanarReach}.
The notion of ``separator'' is central to the algorithm.

\begin{definition}
For any undirected graph $G = (V, E)$ and for any constant $\rho$, $0 < \rho < 1$,
a subset of vertices $S$ is called a \newwd{$\rho$-separator}
if (i) removal of $S$ disconnects $G$ into two subgraphs $A$ and $B$, and
(ii) the number of vertices of any component is at most $\rho\cdot |V|$.
The size of separator is the number of vertices in the separator.
\end{definition}

It is well known that every planar graph with $n$ vertices has a $(2/3)$-separator of size $O(\sqrt{n})$ \cite{gazit1987parallel, lipton1979separator},
and we refer an algorithm which obtains such a separator as {\sf Separator}.

Let $G=(V, E)$, $s$ and $t$ be the given input;
that is,
$G$ is a directed graph, and $s$ and $t$ are the start and goal vertices in $V$.
We assume that $t$ is reachable from $s$ in $G$,
and explain that the algorithm confirms it.
We use $\underline{G}$ to denote an underlying undirected graph of $G$.
The algorithm first uses {\sf Separator} to compute a separator $S$ of size $O(\sqrt{n})$ for $\underline{G}$,
and suppose $G$ is divided into two subgraphs $G[V_0]$ and $G[V_1]$ by $S$
($V_0 \cap V_1 = \emptyset$, $V_0 \cup V_1 \cup S = V$).
Let us fix a path $p$ from $s$ to $t$.
The path $p$ is divided into some $k$ subpaths $p_1, p_2, \ldots, p_k$ by $S$.
Note that the end vertex $u_i$ of $p_i$ is on $S$
and whole path $p_i$ is in either one of $G[V_0\cup S]$ and $G[V_1\cup S]$.
Suppose $p_1$ is in $G[V_0\cup S]$.
By searching in $G[V_0 \cup S]$,
we can find $u_1$ is reachable from $s$.
The algorithm remembers it and
searches $G[V_1 \cup S]$ from $u_1$ in the next step.
Then we can find $p_2$, namely $u_2$ is reachable from $u_1$ and $s$.
By repeating this procedure,
we can confirm $u_i$ is reachable from $s$ for any $i$.
More precisely,
for each vertex $v \in S$,
we keep a boolean value which represents reachability from $s$.
In each searching step,
we start the search from vertices in $S$ whose boolean values are true.
We use this reachability algorithm recursively when searching $G[V_b \cup S]$ $(b\in \{0, 1\})$.
Algorithm~\ref{alg:planarreach} is a pseudo code for this algorithm.
In the actual algorithm,
we have to control the recursion more carefully,
but this is enough for explaining where to modify the algorithm for gadget graphs.

\begin{algorithm}
\caption{{\sf PlanarReach}$(G= (V, E), V_s, R[V_s], V_t)$}\label{alg:planarreach}
\begin{algorithmic}[1]
\REQUIRE A planar graph $G$, start vertices $V_s$, a boolean array $R[V_s]$ for $V_s$, end vertices $V_t$.
\ENSURE Return a boolean array $R[V_t]$ for $V_t$. For any $v\in V_t$, $R[v]$ is true if and only if $v$ is reachable from some vertex $u \in V_s$ such that $R[u]$ is true.
\IF{the size of $V$ is small enough}
\STATE use a standard BFS algorithm and compute $R[V_t]$.
\RETURN $R[V_t]$
\ELSE
\STATE Run {\sf Separator} and obtain a separator $S$ ($G$ is divided into $G[V_0]$ and $G[V_1]$).
\STATE $R[S] = {\sf PlanarReach}(G[V_0 \cup S \cup V_s], V_s, R[V_s], S)$
\WHILE {unsearched paths remain}
\STATE $R[S] = {\sf PlanarReach}(G[V_0 \cup S], S, R[S], S)$ 
\STATE $R[S] = {\sf PlanarReach}(G[V_1 \cup S], S, R[S], S)$ 
\ENDWHILE
\RETURN ${\sf PlanarReach}(G[V_1 \cup S \cup V_t], S, R[S], V_t)$
\ENDIF
\end{algorithmic}
\end{algorithm}

\begin{figure}
\begin{center}
\includegraphics[width= 7 cm]{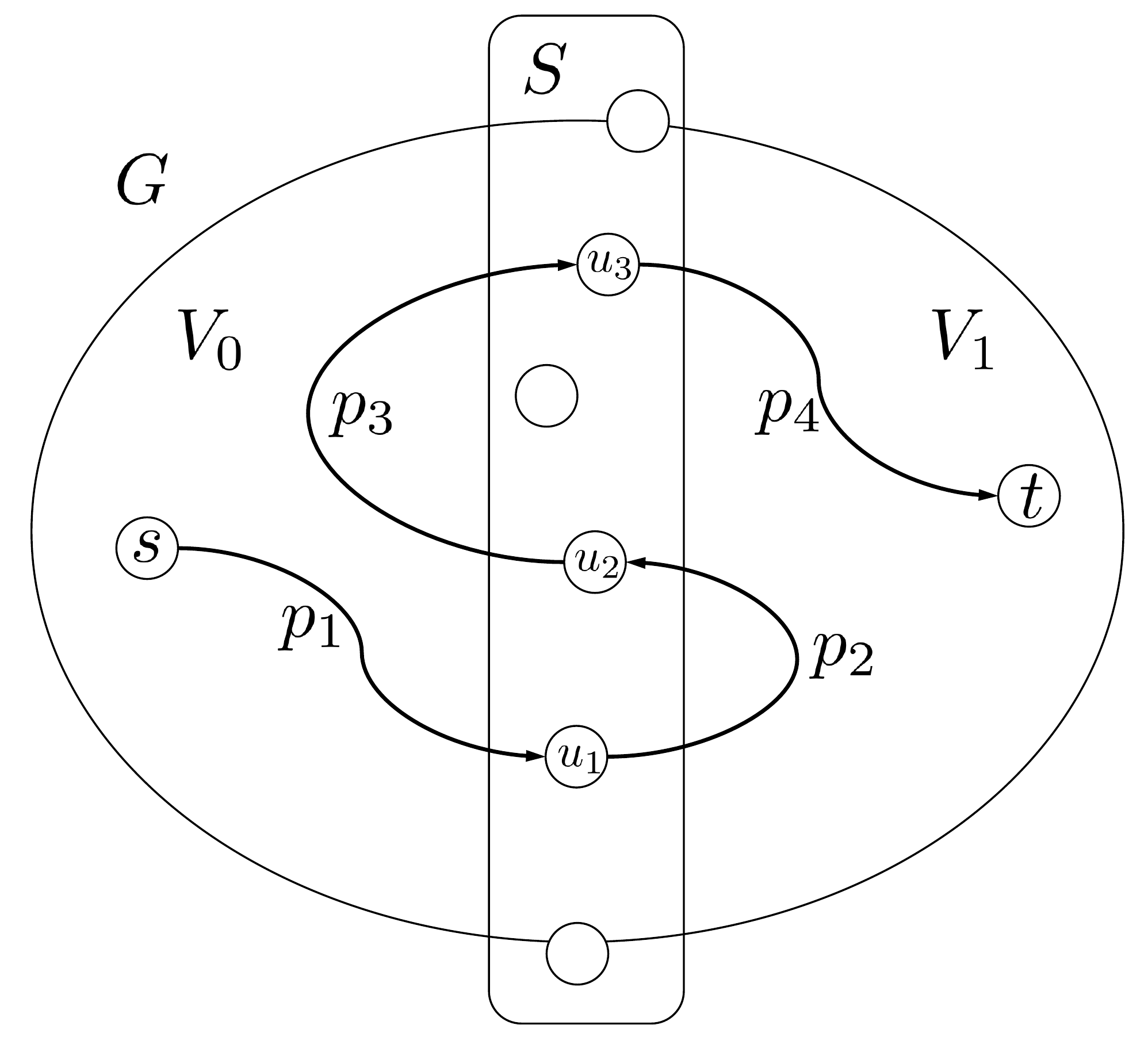}
\caption{An example of a separator $S$ and separated paths.}
\label{fig:alg_sqrt}
\end{center}
\end{figure}

Now,
we explain where to modify.
Let $\gadplG = (\gadplV, \gadplE, \gadplK, \gadplL)$ be an input plane gadget graph of {\sf PlanarReach}
and $N$ be the number of vertices of an input grid graph of Algorithm~\ref{alg:trans_gadget}.
Consider a gadget graph $\gadplG' = (\gadplV', \gadplE', \gadplK', \gadplL')$ which is a subgraph of $\gadplG$.
While we execute {\sf PlanarReach},
for every $v \in S$,
we have to keep a boolean value whether $v$ is reachable from $s$ with using $O(|S|)$ space.
For $\gadplG'$,
instead of the boolean value,
we keep the maximum level that a token starting from $s$ could have at $v$.
When $v$ is equal to $t(\gadplK'(e))$ for some edge $e$,
we should keep a specific level that a token can have at $v$ when the token used the edge $e$ last.
Such a vertex is made by {\sf MakePlanar},
and we should keep at most two specific levels for a vertex.
Thus we use $\widetilde{O}(|S|)$ space for preserving them,
and we can still obtain an $\widetilde{O}(N^{1/3})$ space algorithm.

For $\gadplG'$,
we use Algorithm~\ref{alg:reach_gadget} like Bellman-Ford algorithm instead of BFS.
Algorithm~\ref{alg:reach_gadget} takes as input $\gadplG'$, a start vertex $s$, an initial level $\ell_s$ and an edge restriction $r \in \gadplE \cup \{\bot\}$.
For any $v \in \gadplV'$,
the algorithm computes the maximum level that a token starting from $s$ with a level $\ell_s$ can have at $v$.
When $v$ is equal to $t(\gadplK'(e))$ for some edge $e$,
the algorithm calculates the maximum level that a token can have at $v$ when the token used the edge $e$ last.
In Algorithm~\ref{alg:reach_gadget},
$A[v_e]$ means that
the maximum level
that a token can have at $v$ with using the edge $e$ last,
and
$A[v_{\bot}]$ means that
the maximum level
that a token can have at $v$ with using an edge $e$ last such that $\gadplK'(e) = \bot$.
At the end of $t$-th while-loop,
for any $v\in \gadplV'$,
$A[v_*]$ has the maximum level
which we can have at $v$ within $t$ steps
by starting from $s$ with level $\ell_s$.
At line 4,
we use two mappings $k$ and $\gadK^{-1}$, and they are defined as follows:
\[
	k(e) = \begin{cases}
    \bot & if\  \gadplK'(e) = \bot \\
    e & otherwise
  \end{cases},\ 
  \gadK^{-1}(e) = \begin{cases}
  e' & \exists e',\  \gadplK'(e') = e\\
  \bot & otherwise
  \end{cases}
\]
Since the value $A[\cdot]$ changes no more than two times with the same label,
the while-loop will terminate in $|\gadplL'|$ steps
where $|\gadplL'| = |\bigcup_{e\in \gadplE'} \gadplL'(e)|$. 
Thus the computation time for Algorithm \ref{alg:reach_gadget} is $O(|\gadplL'|^2)$.
An edge has at most $O(N^{1/3})$ labels,
thus the algorithm runs in polynomial time of $N$.

\begin{algorithm}
\caption{}\label{alg:reach_gadget}
\begin{algorithmic}[1]
\REQUIRE A gadget graph $\gadplG'=(\gadplV', \gadplE', \gadplK', \gadplL')$, start vertex $s$, initial level $\ell_s$, edge restriction $r \in \gadplE \cup \{\bot\}$.
\STATE initialize $A[v_\bot] = A[v_e] = -1$ for every $v\in \gadplV'$ and $e\in \gadplE'$ such that $h(e) = v$ except for $s$ and let $A[s_r] = \ell_s$
\WHILE {$A$ was changed in the previous loop}
\FORALL {$e \in \gadE$}
\STATE $A[h(e)_{k(e)}] \leftarrow \max(A[h(e)_{k(e)}], \max\{b\ |\ a\rightarrow b \in \gadplL'(e),\ A[t(e)_{\gadK^{-1}(e)}] \ge a\})$
\ENDFOR
\ENDWHILE
\STATE output $A$
\end{algorithmic}
\end{algorithm}

\section{Conclusion}
We presented an $\widetilde{O}(n^{1/3})$ space algorithm for the grid graph reachability problem.
The most natural question is whether we can apply our algorithm to the planar graph reachability problem.
Although the directed planar reachability is reduced to the directed reachability on grid graphs \cite{allender2009planar},
the reduction blows up the size of the graph by a large polynomial factor and hence it is not useful.
Moreover,
it is known that there exist planar graphs that require quadratic grid area for embedding \cite{valiant1981universality}.
However we do not have to stick to grid graphs.
We can apply our algorithm to graphs which can be divided into small blocks efficiently.
For instance we can use our algorithm for king's graphs \cite{chang2013algorithmic}.
More directly,
for using our algorithm, it is enough to design an algorithm that divides a planar graph into small blocks efficiently.


\begin{thebibliography}{10}

\bibitem{allender2009planar}
Eric Allender, David A~Mix Barrington, Tanmoy Chakraborty, Samir Datta, and
  Sambuddha Roy.
\newblock Planar and grid graph reachability problems.
\newblock {\em Theory of Computing Systems}, 45(4):675--723, 2009.

\bibitem{asano2011memory}
Tetsuo Asano and Benjamin Doerr.
\newblock Memory-constrained algorithms for shortest path problem.
\newblock In {\em CCCG}, 2011.

\bibitem{asano2014widetilde}
Tetsuo Asano, David Kirkpatrick, Kotaro Nakagawa, and Osamu Watanabe.
\newblock $\widetilde{O}(\sqrt{n})$-space and polynomial-time algorithm for
  planar directed graph reachability.
\newblock In {\em International Symposium on Mathematical Foundations of
  Computer Science}, pages 45--56. Springer, 2014.

\bibitem{barnes1998sublinear}
Greg Barnes, Jonathan~F Buss, Walter~L Ruzzo, and Baruch Schieber.
\newblock A sublinear space, polynomial time algorithm for directed st
  connectivity.
\newblock {\em SIAM Journal on Computing}, 27(5):1273--1282, 1998.

\bibitem{bourke2009directed}
Chris Bourke, Raghunath Tewari, and NV~Vinodchandran.
\newblock Directed planar reachability is in unambiguous log-space.
\newblock {\em ACM Transactions on Computation Theory (TOCT)}, 1(1):4, 2009.

\bibitem{chang2013algorithmic}
Gerard~Jennhwa Chang.
\newblock Algorithmic aspects of domination in graphs.
\newblock {\em Handbook of Combinatorial Optimization}, pages 221--282, 2013.

\bibitem{gazit1987parallel}
Hillel Gazit and Gary~L Miller.
\newblock A parallel algorithm for finding a separator in planar graphs.
\newblock In {\em Foundations of Computer Science, 1987., 28th Annual Symposium
  on}, pages 238--248. IEEE, 1987.

\bibitem{imai2013n12+}
Tatsuya Imai, Kotaro Nakagawa, Aduri Pavan, NV~Vinodchandran, and Osamu
  Watanabe.
\newblock An ${O}(n^{1/2+\varepsilon})$-space and polynomial-time algorithm for
  directed planar reachability.
\newblock In {\em Computational Complexity (CCC), 2013 IEEE Conference on},
  pages 277--286. IEEE, 2013.

\bibitem{lipton1979separator}
Richard~J Lipton and Robert~Endre Tarjan.
\newblock A separator theorem for planar graphs.
\newblock {\em SIAM Journal on Applied Mathematics}, 36(2):177--189, 1979.

\bibitem{reingold2008undirected}
Omer Reingold.
\newblock Undirected connectivity in log-space.
\newblock {\em Journal of the ACM (JACM)}, 55(4):17, 2008.

\bibitem{stolee2012space}
Derrick Stolee and NV~Vinodchandran.
\newblock Space-efficient algorithms for reachability in surface-embedded
  graphs.
\newblock In {\em Computational Complexity (CCC), 2012 IEEE 27th Annual
  Conference on}, pages 326--333. IEEE, 2012.

\bibitem{valiant1981universality}
Leslie~G Valiant.
\newblock Universality considerations in vlsi circuits.
\newblock {\em IEEE Transactions on Computers}, 100(2):135--140, 1981.

\bibitem{wigderson1992complexity}
Avi Wigderson.
\newblock The complexity of graph connectivity.
\newblock {\em Mathematical Foundations of Computer Science 1992}, pages
  112--132, 1992.

\end{thebibliography}
\end{document}